\documentclass[runningheads]{llncs}

\usepackage{xspace}
\usepackage[utf8]{inputenc}
\usepackage{graphicx}
\usepackage[T1]{fontenc} 
\usepackage{url}
\usepackage{ifthen}
\usepackage{enumitem,wrapfig}
\usepackage[cmex10]{amsmath}
\usepackage{amssymb}
\usepackage{cite}

\usepackage{amsthm}
\usepackage{textcomp}
\usepackage{siunitx}
\usepackage{color}
\usepackage{float}
\usepackage[lofdepth,lotdepth]{subfig}
\usepackage[countmax]{subfloat}
\usepackage{cases}
\usepackage[font={small}]{caption}
\usepackage{tcolorbox}
\usepackage{fix-cm}
\usepackage{bbm}
\usepackage{hyperref}
\usepackage{algorithm}
\usepackage[noend]{algpseudocode}

\usepackage{xcolor}
\definecolor{azure}{rgb}{0.54, 0.17, 0.89}
\newcommand{\colorcomment}[1]{\Comment{ {\color{azure} #1}} }
\newcommand{\maincolorcomment}[1]{{\color{azure}// #1 } }
\makeatletter
\def\thm@space@setup{\thm@preskip=2pt
\thm@postskip=2pt \itshape}
\makeatother
\newtheoremstyle{newstyle}      
{} 
{} 
{\mdseries} 
{} 
{\bfseries} 
{.} 
{ } 
{} 

\theoremstyle{definition}

\theoremstyle{remark}

\newtheorem{defn}{Definition}

\makeatletter
\newtheorem*{rep@theorem}{\rep@title}
\newcommand{\newreptheorem}[2]{%
\newenvironment{rep#1}[1]{%
 \def\rep@title{#2 \ref{##1}}%
 \begin{rep@theorem}}%
 {\end{rep@theorem}}}
\makeatother

\newreptheorem{theorem}{\textbf{Theorem}}

\newcommand{\N}{{\mathcal{N}}}

\renewcommand{\H}{{\mathcal{H}}}
\newcommand{\A}{{\mathcal{A}}}
\newcommand{\C}{{\mathcal{C}}}

\newcommand{\blue}{\textcolor{blue}}

\newcommand{\Z}{{\mathcal{Z}}}

\newcommand{\T}{{\mathcal{T}}}
\newcommand{\cF}{{\mathcal{F}}}
\newcommand{\cN}{{\mathcal{N}}}

\renewcommand{\tt}{\mathcal{T}}
\newcommand{\zz}{\mathcal{Z}}

\newcommand{\I}{{\mathcal I}}
\newcommand{\E}{{\mathcal E}}
\newcommand{\W}{{\mathcal W}}
\newcommand{\J}{{\mathcal J}}

\newcommand{\prot}{\mathrm{\Pi}}
\newcommand{\EXEC}{\textnormal{\textsf{EXEC}}}
\newcommand{\secparam}{\kappa}
\newcommand{\adv}{\mathcal{A}}
\newcommand{\env}{\mathcal{Z}}
\newcommand{\negl}{\textnormal{\textsf{negl}}}

\newcommand{\view}{\textnormal{\textsf{view}}}
\newcommand{\getsr}{{\:{\leftarrow{\hspace*{-3pt}\raisebox{.75pt}{$\scriptscriptstyle\$$}}}\:}}
\newcommand{\killsignal}{\textnormal{\texttt{kill}}}
\newcommand{\corrupt}{\textnormal{\texttt{corrupt}}}
\newcommand{\pp}{\textbf{pp}}
\newcommand{\G}{\mathcal G}
\newcommand{\rs}{\texttt{randSource}}
\newcommand{\vdf}{\textsc{VDF}}
\newcommand{\rvdf}{\textsc{RandVDF}}
\newcommand{\rvdfeval}{\textsc{RandVDF.Eval}}
\newcommand{\rvdfverify}{\textsc{RandVDF.Verify}}

\newcommand{\stake}{\texttt{stake}}

\newcommand{\parentBk}{\texttt{parentBlk}}
\newcommand{\unCnfTx}{\texttt{unCnfTx}}
\newcommand{\op}{\texttt{output}}
\newcommand{\ip}{\texttt{input}}
\newcommand{\pf}{\texttt{proof}}
\newcommand{\s}{\texttt{s}}
\newcommand{\randiter}{\texttt{randIter}}
\newcommand{\state}{\texttt{state}}
\newcommand{\content}{\texttt{content}}
\newcommand{\coin}{\texttt{coin}}
\newcommand{\Time}{\texttt{slot}}

\makeatletter
\renewcommand*{\@fnsymbol}[1]{\ensuremath{\ifcase#1\or { }\or \dagger\or \ddagger\or
    \mathsection\or \mathparagraph\or \|\or **\or \dagger\dagger
    \or \ddagger\ddagger \else\@ctrerr\fi}}
\makeatother

\newcommand{\vb}[1] {{\textcolor{azure}{}}}

\newcommand{\protocol}{{\sf   PoSAT}\xspace}

\title{PoSAT: Proof-of-Work Availability and Unpredictability, without the Work}

\author{
Soubhik Deb$^\ddagger$,
Sreeram Kannan$^\ddagger$,
David Tse$^\star$\\
Email: soubhik@uw.edu,
ksreeram@uw.edu, 
dntse@stanford.edu 
}

\institute{
$^\ddagger$University of Washington,\\
$^\star$Stanford University}

\begin{document}
\sloppy

\maketitle
\begin{abstract}
    An important feature of Proof-of-Work (PoW) blockchains is full dynamic availability, allowing miners to go online and offline while requiring only $50\%$ of the online miners to be honest. 
    Existing Proof-of-stake (PoS), Proof-of-Space and related protocols are able to achieve this property only partially, either requiring the additional assumption that adversary nodes are online from the beginning and no new adversary nodes come online afterwards, or use additional trust assumptions for newly joining nodes.   We propose a new PoS protocol \protocol which can provably achieve dynamic availability fully without any additional assumptions. The protocol is based on the longest chain and uses a Verifiable Delay Function for the block proposal lottery to provide an arrow of time. The security analysis of the protocol draws on the recently proposed technique of Nakamoto blocks as well as the theory of branching random walks. An additional feature of \protocol is the complete unpredictability of who will get to propose a block next, even by the winner itself. This unpredictability is at the same level of PoW protocols, and  is stronger than that of existing PoS protocols using Verifiable Random Functions.  

\end{abstract}


\section{Introduction}
\label{sec:intro}

\subsection{Dynamic Availability}

Nakamoto's invention of Bitcoin \cite{bitcoin} in 2008 brought in the novel concept of a permissionless Proof-of-Work (PoW) consensus protocol. Following the longest chain protocol, a block can be proposed and appended to the tip of the blockchain if the miner is successful in solving the hash puzzle.  The Bitcoin protocol has several interesting features as a consensus protocol. An important one is {\em dynamic availability}.  Bitcoin can handle an uncertain and dynamic varying level of consensus participation in terms of mining power. Miners can join and leave as desired without any registration requirement. This is in contrast to most classical Byzantine Fault Tolerant (BFT) consensus protocols, which  assumes a fixed and known number of consensus nodes. Indeed, Bitcoin has been continuously available since the beginning, a period over which the hashrate has varied over a range of $14$ orders of magnitude. Bitcoin has been proven to be secure as long as the attacker has less than $50\%$ of the online hash power (the static power case is considered in \cite{bitcoin,backbone,pss16} and variable hashing power case is considered in \cite{nakamoto_bounded_delay,backbone_var_difficulty}).
    
    
Recently proof-of-stake (PoS)  protocols have emerged as an energy-efficient alternative to PoW. Instead of solving a difficult hash puzzle, nodes participate in a lottery  to win the right to append a block to the blockchain, with the probability of winning  proportional to a node's stake  in the total pool.   
This replaces the resource intense mining process of PoW, while ensuring fair chances to contribute and claim rewards. 

There are broadly two classes of PoS protocols: those derived from classical BFT protocols and those  inspired by Nakamoto's longest chain protocol. Attempts at blockchain design via the BFT approach include Algorand \cite{chen2016algorand,gilad2017algorand}, Tendermint \cite{tendermint}  and  Hotstuff \cite{yin2018hotstuff}. Motivated and inspired by Nakamoto longest chain protocol are the PoS designs of  Snow White \cite{bentov2016snow} and the Ouroboros family  of protocols \cite{kiayias2017ouroboros,david2018ouroboros,badertscher2018ouroboros}. One feature that distinguish the  PoS longest chain protocols from the BFT protocols is that they inherit the dynamic availability of Bitcoin: the chain always grows regardless of the number of nodes online.  But do these PoS longest chain protocols provide the same level of security guarantee as PoW Bitcoin in the dynamic setting?

\subsection{Static vs Dynamic Adversary}

Two particular papers focus on the problem of dynamic availability in PoS protocols: the sleepy model of consensus \cite{sleepy} and Ouroboros Genesis \cite{badertscher2018ouroboros}. In both papers, it was proved that their protocols are secure if less than $50\%$ of the online nodes are adversary. This condition is the same as the security guarantee in PoW Bitcoin, but there is an additional assumption: {\em all adversary nodes are always online starting from genesis and no new adversary nodes can join.}  While this static adversary assumption seems reasonable  (why would an adversary go to sleep?), in reality this can be a very restrictive  condition. In the context of Bitcoin, this assumption would be analogous to the statement that the hash power of the adversary is fixed in the past decade ( while the total hashing power increased $14$ orders of magnitude!) 
More generally, in public blockchains, PoW or PoS, no node is likely to be adversarial during the launch of a new blockchain token - adversaries only begin to emerge later during the lifecycle. 


The static adversary assumption underlying these PoS protocols is not superfluous but is in fact {\em necessary} for their security. Suppose for the $1^{st}$ year of the existence of the PoS-based blockchain, only $10\%$ of the total stake is online. Out of this, consider that all nodes are honest. Now, at the beginning of the $2^{nd}$ year, all $100\%$ of the stake is online out of which $20\%$ is held by adversary. At any point of time, the fraction of  online stake held by honest nodes is greater than $0.8$. However, both Sleepy and Genesis are not secure since the adversary can use its $20\%$ stake to immediately participate in  all past lotteries to win blocks all the way back to the genesis and then grow a chain {\em instantaneously} from the genesis to surpass the current longest chain (Figure \ref{fig:aot}(a)). Thus, due to this ``costless simulation'', newly joined adversary nodes not only increase the current online adversary stake, but effectively increase past online adversary stake as well. See Appendix~\ref{sec:costless_simulation_attack} for further details on how costless simulation renders both sleepy model of consensus and Ouroboros Genesis vulnerable to attacks. In contrast, PoW does not suffer from the same issue because it would take a long time to grow such a chain from the past and that chain will always be behind the current longest chain. Thus, PoW provides an {\em arrow of time}, meaning nodes cannot ``go back in time'' to mine blocks for the times at which they were not online. This property is key in endowing PoW protocols with the ability to tolerate fully dynamic adversaries wherein  both  honest  nodes  and  adversary  can have  varying  participation (Figure \ref{fig:aot}(b)).

\begin{figure*}
    \centering
    \includegraphics[width=\textwidth]{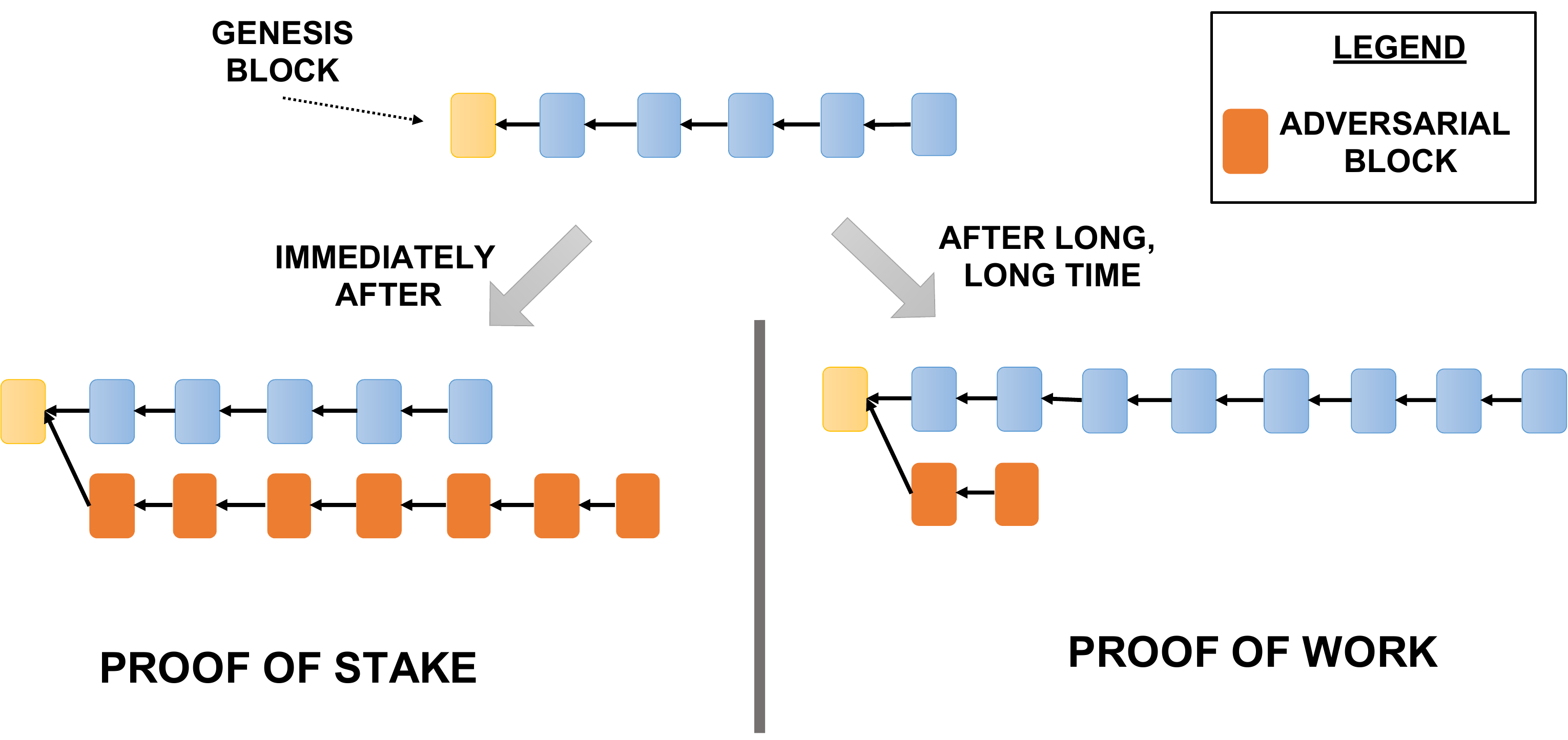}
    \caption{(a) Newly joined nodes in existing PoS protocols can grow a chain from genesis instantaneously. (b) Newly joined miners in PoW protocol takes a long time to grow such a chain and is always behind.}
    \label{fig:aot}
\end{figure*}

We point out that some protocols including Ouroboros Praos \cite{david2018ouroboros} and Snowhite \cite{bentov2016snow} require that nodes discard chains that fork off too much from the present chain. This feature was introduced to handle nodes with expired stake (or nodes that can perform key grinding) taking over the longest chain. While they did not specifically consider the dynamic adversary issue we highlighted, relying on previous checkpoints can potentially solve the aforementioned security threat. However, as was eloquently argued in Ouroboros Genesis \cite{badertscher2018ouroboros},  these checkpoints are unavailable to offline clients and newly joining nodes require advice from a trusted party (or a group inside which a majority is trusted). This trust assumption is too onerous to satisfy in practice and is not required in PoW. Ouroboros Genesis was designed to require no trusted joining assumption while being secure to long-range and key-grinding attacks. However, they are not secure against dynamic participation by the adversary: they are vulnerable to the aforementioned attack.  This opens the following question: 

{\em Is there a fully dynamically available PoS protocol which has full PoW  security guarantee, without additional trust assumptions? }

\subsection{\protocol achieves PoW dynamic availability}


We answer the aforementioned question in the affirmative.  Given that arrow-of-time is a central property of PoW protocols, we design a new PoS protocol, {\sf PoS with Arrow-of-Time} (\protocol),  also having this property using randomness generated from Verifiable Delay Functions (VDF).  VDFs are built on top of iteratively sequential functions, i.e., functions that are only computable sequentially: $f^{\ell}(x) = f \circ f \circ ... \circ f(x)$, along with the ability to provide a short and easily verifiable proof that the computed output is correct. Examples of such functions include (repeated) squaring in a finite group of unknown order \cite{cai1993towards,rivest1996time}, i.e,, $f(x)=2^x$ and (repeated) application of secure hash function (SHA-256) \cite{mahmoody2013publicly}, i.e,, $f(x) = \textsc{Hash}(x)$. While VDFs have been designed as a way for proving the passage of a certain amount of time (assuming a bounded CPU speed), it has been recently shown that these functions can also be used to generate an unpredictable randomness beacon \cite{ephraim2020continuous}. Thus, running the iteration till the random time $L$ when $\rvdf(x) = f^{L}(x) < \tau$ is within a certain threshold will result in $L$ being a geometric random variable. We will incorporate this randomized VDF functionality to create an arrow-of-time in our protocol.

The basic idea of our protocol is to mimic the PoW lottery closely: instead of using the solution of a Hash puzzle based on the parent block's hash as proof of work, we instead use the randomized VDF computed based on the parent block randomness and the coin's public key as the proof of stake lottery. In a PoW system,  we are required to find a string called "\texttt{nonce}" such that $\textsc{Hash}(\texttt{block}, \texttt{nonce}) < \tau$, a hash-threshold. Instead in our PoS system, we require $\rvdf(\rs, pk, \Time) < \tau$, where $\rs$ is the randomness from the parent block, $pk$ is the public key associated with the mining coin and $\Time$ represents the number of iterations of the $\rvdf$ since genesis. There are four differences, the first three are common in existing PoS systems: (1) we use ``$\rs$" instead of ``\texttt{block}" in order to prevent grinding attacks on the content in the PoS system, (2) we use the public-key ``pk" of staking coin instead of PoW ``\texttt{nonce}" to simulate a PoS lottery, (3) we use ``$\Time$" for ensuring time-ordering,
(4) instead of using a \textsc{Hash}, we use the $\rvdf$, which requires sequential function evaluation thus creating an ``arrow of time".

The first two aspects are common to many PoS protocols and is most similar to an earlier PoS protocol \cite{fan2018scalable}, however, crucially we use the $\rvdf$ function instead of a Verifiable random function (VRF) and a time parameter inside the argument used in that protocol. This change allows for full dynamic availability: if adversaries join late, they cannot produce a costless simulation of the time that they were not online and build a chain from genesis instantaneously. It will take the adversary time to grow this chain (due to the sequential nature of the $\rvdf$), by which time, the honest chain would have grown and the adversary will be unable to catch up. Thus, \protocol behaves more like PoW (Figure \ref{fig:aot}(b)) rather than existing PoS based on VRF's (Figure \ref{fig:aot}(a)).  We show that this protocol achieves full dynamic availability: if $\lambda_h(t)$ denotes the honest stake online at $t$, $\lambda_a(t)$ denotes the online adversarial stake at time $t$, it is secure as long as 
\begin{align}
\label{eq:informal}
    \lambda_h(t) > e \lambda_a(t) \qquad \mbox{for all $t$},
\end{align}
where $e$ is Euler's number $2.7182\ldots$. 

We observe that the security of this protocol requires a stronger condition than PoW protocols. 
The reason for this is that an adversary can potentially do parallel evaluation of VDF on {\em all} possible blocks. Since the randomness in each of the blocks is independent from each other, the adversary has many random chances to increase the chain growth rate to out-compete the honest tree. This is a consequence of the nothing-at-stake phenomenon: the same stake can be used to grind on the many blocks. The factor $e$ is the  resulting  amplification factor for the adversary growth rate.
This is avoided in PoW protocols due to the conservation of work inherent in PoW which requires the adversary to split its total computational power among such blocks. 

\begin{figure}
    \centering
    \includegraphics[scale=0.25]{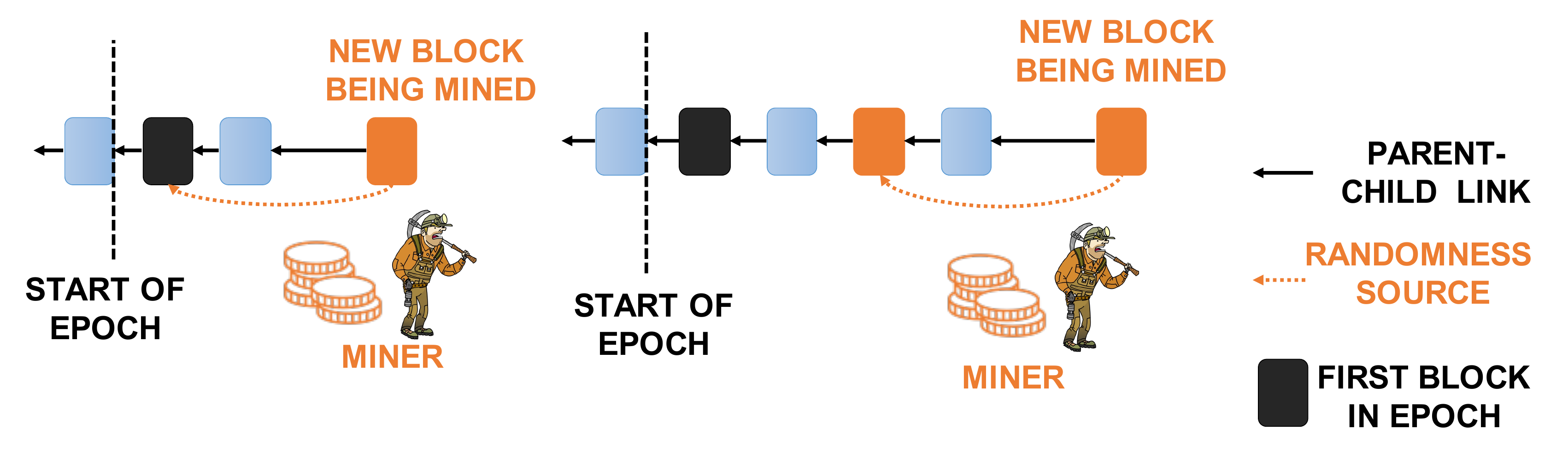}
    \caption{Left: A node uses randomness from the first block of the epoch. Right: Since a node  already won a block in the period, it uses that block's  randomness.}
    \label{fig:mining_link}
\end{figure}

We solve this problem in \protocol by  reducing the rate at which the block randomness is updated and hence reducing the block randomness grinding opportunities of the adversary. Instead of updating the block randomness at every level of the blocktree, we only update it once every $c$ levels (called an epoch). The larger the value of the parameter $c$, the slower the block randomness is updated.  The common source of randomness used to run the VDF lottery  remains the same for $c$ blocks starting from the genesis and is
updated only  when (a) the current block to be generated is at a depth that is a multiple of $c$, or (b) the coin used for the lottery is successful within the epoch of size $c$. The latter condition is necessary to create further independent winning opportunities for the node within the period $c$ once a slot is obtained with that coin. This is illustrated in Figure~\ref{fig:mining_link}. For $c=1$, this corresponds to the protocol discussed earlier. 


The following security theorem is proved about \protocol for general $c$, giving a condition for security (liveness and persistence) under {\em all} possible attacks .

\begin{reptheorem}{thm:main}[Informal]
\protocol with parameter $c$ is secure as long as  
\begin{equation}
\label{eq:cond_thm1}
\frac{\lambda^c_h(t)}{1+\lambda_{\max}\Delta} > \phi_c\lambda_a(t)  \qquad  \mbox{for all $t$},
\end{equation}
where $\lambda^c_h(t)$ is the honest stake this is online at time $t$ and has been online since at least $t-\Theta(c)$, $\Delta$ is the network delay between honest nodes, $\lambda_{\max}$ is a constant such that $\lambda^c_h(t) \leq \lambda_{\max}$ for all $t>0$, $\phi_c$ is a constant, dependent on $c$,  given in (\ref{eq:phi_c}). $\phi_1 = e$ and $\phi_c \rightarrow 1$ as $c \rightarrow \infty$.
 \end{reptheorem}
 
 \begin{table}[h]
\begin{center}
\begin{tabular}{ |c|c|c|c|c|c|c|c|c|c|c| }
 \hline
 $c$      & 1 & 2 & 3 & 4 & 5 & 6 & 7 & 8 & 9 & 10\\ \hline
 $\phi_c$   & e & 2.22547  & 2.01030 & 1.88255  & 1.79545  & 1.73110 & 1.68103 & 1.64060 &1.60705 & 1.57860 \\
 \hline
 $\frac{1}{1+\phi_c}$ & $\frac{1}{1+e}$& 0.31003 & 0.33219 &0.34691 &0.35772 & 0.36615 & 0.37299 &
 0.37870 & 0.38358 & 0.38780
 \\ \hline 
\end{tabular}
\end{center}
\caption{Numerically computed values of the adversary amplification factor $\phi_c$. The ratio $1/(1+\phi_c)$ is the adversarial fraction of stake that can be tolerated by \protocol when $\Delta = 0$.}
\vspace{-1cm}
\end{table}
 
We remark that in our PoS protocol, we have a known upper bound on the rate of mining blocks (by assuming that the entire stake is online). We can use this information to set $1+\lambda_{\max} \Delta$ as close to $1$ as desired by simply setting the mining threshold appropriately. Furthermore, by setting $c$ large, $\phi_c \approx 1$ and thus \protocol can achieve the same security threshold as PoW under full dynamic availability.  
 The constant $\phi_c$ is the amplification of the adversarial chain growth rate due to nothing-at-stake, which we calculate using the theory of branching random walks\cite{shi}. The right hand side of (\ref{eq:cond_thm1}) can therefore be interpreted as the growth rate of a private adversary tree with the adversary mining on every block.  Hence, condition (\ref{eq:cond_thm1}) can be interpreted 
 as the condition that the private Nakamoto attack \cite{bitcoin} does not succeed. However, Theorem \ref{thm:main} is a {\em security theorem}, i.e. it gives a condition under which the protocol is secure under {\em all} possible attacks.  Hence what Theorem \ref{thm:main} says is therefore that among all possible attacks on \protocol, the private attack is the {\em worst} attack. We prove this by using the technique of blocktree partitioning and Nakamoto blocks, introduced in \cite{dembo2020everything}, which reduce all attacks to a union of private attacks.

We note that large $c$ is beneficial from the point of view of getting a tight security threshold. However, we do require $c$ to be finite (unlike other protocols like Ouroboros that continue to work under $c$ being infinite). This is because the latency to confirm a transaction increases linearly in $c$ (see Section~\ref{sec:analysis}). Furthermore, an honest node on coming online has to wait until encountering the next epoch beginning before it can participate in proposing blocks and the worst-case waiting time  increases linearly with $c$. We note that the adversary cannot use the stored blocks in the next epoch, thus having a bounded reserve of blocks. The total number of blocks stored up by an adversary potentially increases linearly in the epoch size, thus requiring the  confirmation depth and thus latency to be larger than $\Theta(c)$.  By carefully bounding this enhanced power of the adversary, for any finite $c$, we show that \protocol is secure.

Assuming $\lambda_{\max}\Delta$ to be small, the comparison of \protocol with other protocols is shown in Table~\ref{tab:1}. Here we use  $\Lambda_a$ to be the largest adversary fraction of the total stake online at any time during the execution ($\Lambda_a = \sup_t \lambda_a(t)$). Protocols whose security guarantee assumes all adversary nodes are online all the time effectively assumes that $\lambda_h(t) > \Lambda_a$. Thus existing protocols  have limited dynamic availability.

\begin{center}
 \begin{tabular}{||c | c | c | c | c | c ||} 
 \hline
 & Ourboros & Snow White /  & Genesis / & Algorand &  \protocol  \\
 &  & Praos & Sleepy & &   \\
 \hline\hline
 Dynamic &  $\lambda_h(t) > \Lambda_a$ & $\lambda_h(t) > \Lambda_a$  & $\lambda_h(t) > \Lambda_a$ & No   &   $\lambda^c_h(t) >  \phi_c \lambda_a(t)$ \\ 
 Availability &  &  & & &  \\ 
 \hline
 Predictability &  Global  & Local & Local & Local & None \\
 \hline
\end{tabular} \label{tab:1}
\end{center}

\subsection{\protocol has PoW Unpredictability}

Another key property of PoW protocols is their ability to be unpredictable: no node (including itself) can know when a given node will be allowed to propose a block ahead of the proposal slot. We point out that $\protocol$ with any parameter $c$ remains unpredictable due to the the unpredictability of the RandVDF till the threshold is actually reached. We refer the reader to Fig.~\ref{fig:mining_link}(a) where if the randomness source is at the beginning of the epoch it is clear that the unpredictability of the randomized VDF implies unpredictability in our protocol. However, in case the miner has already created a block within the epoch (Fig.~\ref{fig:mining_link}(b)), the randomness source is now her previous block. This can be thought of as a continuation of the iterative sequential function  from the beginning of the epoch and hence it is also unpredictable as to when the function value will fall below a threshold. Thus \protocol achieves true unpredictability, matching the PoW gold standard, where even an all-knowing adversary has no additional predictive power. 

The first wave of PoS protocols such as Ouroboros \cite{kiayias2017ouroboros}  are fully predictable as they rely on mechanisms for proposer election that provide global knowledge of all proposers in an epoch ahead of time. The concept of Verifiable Random Functions (VRF), developed in \cite{vrf,vrf2}, was pioneered in the blockchain context in Algorand \cite{chen2016algorand,gilad2017algorand}, as well as applied in Ouroboros Praos \cite{david2018ouroboros} and Snow White \cite{bentov2016snow}. The use of a private leader election using VRF enables no one else other than the proposer to know of the slots when it is allowed to propose blocks.  However, unlike Bitcoin, the proposer itself can predict. Thus, these protocols still allow {\em local} predictability. 
The following vulnerability is caused by local predictability: a rational node may then willingly sell out his slot to an adversary. In Ouroboros Praos, such an all-knowing adversary needs to corrupt only $1$ user at a time (the proposer) adaptively in order to do a double-spend attack. He will first let the chain build for some time to confirm a transaction, and then get the bribed proposers one at a time to build a competing chain. Algorand is more resilient, but even there, in each step of the BFT algorithm, a different committee of nodes is selected using a VRF based sortition algorithm. These nodes are locally predictable as soon as the previous block is confirmed by the BFT - and thus an all-knowing adversary only needs to corrupt a third of a committee. Assuming each committee is comprised of $K$ nodes ($K$  being a constant), the adversary only needs to corrupt $\frac{K}{3N}$ fraction of the nodes. Refer to Appendix~\ref{sec:bribery_attack} for further details.

We summarize the predictability of various protocols in Table~\ref{tab:1}.

\subsection{Related Work}

Our design is based on frequent updates of randomness to run the VDF lottery. PoS protocols that update randomness at each iteration have been utilized in practice as well as theoretically proposed \cite{fan2018scalable} - they do not use VDF and have {\bf neither} dynamic availability nor unpredictability. Furthermore, they still face nothing-at-stake attacks. In fact, the amplification factor of $e$ we discussed earlier has been first observed in a Nakamoto private attack analysis in \cite{fan2018scalable}. This analysis was subsequently extended to a full security analysis against {\em all}
attacks in \cite{pos_revisited_arxiv,dembo2020everything}, where it was shown that the private attack is actually the worst attack. In \cite{pos_revisited_arxiv}, the idea of $c$-correlation was introduced to reduce the rate of randomness update and to reduce the severity of the nothing-at-stake attack; we borrowed this idea from them in the design of our VDF-based protocol, \protocol. 

There have been attempts to integrate VDF into the proof-of-space paradigm \cite{cohen2019chia} as well as into the proof-of-stake paradigm \cite{azouvi2018betting}, \cite{long2019nakamoto}, all using a VRF concatenated with a VDF. But, in \cite{cohen2019chia}, the VDF runs for a fixed duration depending on the input and hence is predictable, and furthermore do not have security proofs for dynamic availability. In \cite{azouvi2018betting}, the randomness beacon is not secure till the threshold of $1/2$ as claimed by the authors since it  has a randomness grinding attack which can potentially expand the adverarial power by at least factor $e$. There are three shortcomings in \cite{long2019nakamoto} as compared to our paper: (1) even under static participation, they only focus on an attack where an adversary grows a private chain, (2) there is no modeling of dynamic availablility and a proof of security and (3) since the protocol focuses only on $c=1$, they can only achieve security till threshold $1/{1+e}$, not till $1/2$. We note that recent work \cite{brown2019formal} formalized that a broad class of PoS protocols suffer from either of the two vulnerabilities: (a) use recent randomness, thus being subject to nothing-at-stake attacks or (b) use old randomness, thus being subject to prediction based attacks (even when only locally predictable). We note that \protocol with large $c$ completely circumvents both vulnerabilities using the additional VDF primitive since it is able to use old randomness while still being fully unpredictable.

We want to point out that dynamic availability is distinct and complementary to {\em dynamic stake}, which implies that the set of participants and their identities in the mining is changing based on the state of the blockchain. We note that there has been much existing work addressing issues on the dynamic stake setting - for example, the $s$-longest chain rule in \cite{badertscher2018ouroboros}, whose adaptation to our setting we leave for future work. We emphasize that the dynamic availability problem is well posed even in the static stake setting (the total set of stakeholders is fixed at genesis).

\subsection{Outline}

The rest of the paper is structured as follows. Section \ref{sec:protocols} presents the VDF primitive we are using and the overall protocol. Section \ref{sec:model} presents the model. Section \ref{sec:analysis} presents the details of the security analysis.

\section{Protocol}
\label{sec:protocols}
\subsection{Primitives}
\label{sec:primitives}
In this section, we give an overview of $\vdf$s and refer the reader  to detailed definitions in Appendix~\ref{sec:vdf}.
\begin{definition}[from \cite{boneh2018verifiable}]
A $\vdf$ $V = (\textsc{Setup}, \textsc{Eval}, \textsc{Verify})$ is a triple of algorithms as follows:
\begin{itemize}
    \item $\textsc{Setup}(\lambda,\tau) \rightarrow \mathbf{pp} = (ek, vk)$ is a randomized algorithm that produces an evaluation
    key $ek$ and a verification key $vk$. 
    \item $\textsc{Eval}(ek, input, \tau) \rightarrow (O, proof)$ takes an  $input \in \mathcal{X}$, an evaluation key $ek$, number of steps $\tau$ and produces an output $O \in \mathcal{Y}$ and a (possibly empty) $proof$. 
    \item $\textsc{Verify}(vk, input, O, proof,\tau) \rightarrow {Yes, No}$ is a deterministic algorithm takes an input, output, proof, $\tau$ and outputs $Yes$ or $No$. 
\end{itemize}
\end{definition}
$\vdf$.\textsc{Eval} is usually comprised of sequential evaluation: $f^{\ell}(x) = f \circ f \circ ... \circ f(x)$  along with the ability to provide a short and easily verifiable proof. In particular, there are three separate functions $\vdf$.\textsc{Start}, $\vdf$.\textsc{Iterate} and $\vdf$.\textsc{Prove} (the first function is used to initialize,  the second one operates for the number of steps and the third one furnishes a proof). This is illustrated in Figure~\ref{fig:traditional_VDF} on the left.
While $\vdf$s have been designed as a way for proving the passage of a certain amount of time, it has been recently shown that these functions can also be used to generate an unpredictable randomness beacon \cite{ephraim2020continuous}. Thus, running the iteration till the random time $L$ when $\rvdf(x) = f^{L}(x) < \tau$ generates the randomness beacon. This is our core transformation to get a randomized $\vdf$. This is shown in Figure~\ref{fig:proposed_VDF} on the right. Instead of running for a fixed number of iterations, we run the $\vdf$ iterations till it reaches a certain threshold. Our transformation is relatively general purpose and most $\vdf$s can be used with our construction. For example, a $\vdf$ (which is based on squaring in a group of unknown order) is an ideal example for our construction \cite{pietrzak2018simple,wesolowski2020efficient}. In the recent paper \cite{ephraim2020continuous}, for that sequential function, a new method for obtaining a short proof whose complexity does not depend (significantly) on the number of rounds is introduced - our protocol can utilize that $\vdf$ as well. They show furthermore that they obtain a continuous $\vdf$ property which implies that partial $\vdf$ computation can be continued by a different party - we do not require this additional power in our protocol.

For the $\rvdf$ in $\protocol$, as illustrated in Fig~\ref{fig:proposed_VDF}, $\Time$ plays a similar role as the timestamps in other PoS protocols like \cite{sleepy}. The $\Time$ basically mentions the number of times the $\rvdf$ has iterated since the genesis and when the speed of the iteration of $\rvdf$ is constant, $\Time$ is an approximation to the time elapsed since the beginning of the operation of the PoS system.

Normally, a $\vdf$ will satisfy {\em correctness} and {\em soundness}. And we require $\rvdf$ to also satisfy correctness and soundness as defined in Appendix~\ref{sec:vdf}. 

\begin{figure}[ht!]
    \centering
   \subfloat[\textsc{$\vdf$.Eval}($input,ek,\tau$)  ]{\label{fig:traditional_VDF}
      \includegraphics[width=.34\textwidth]{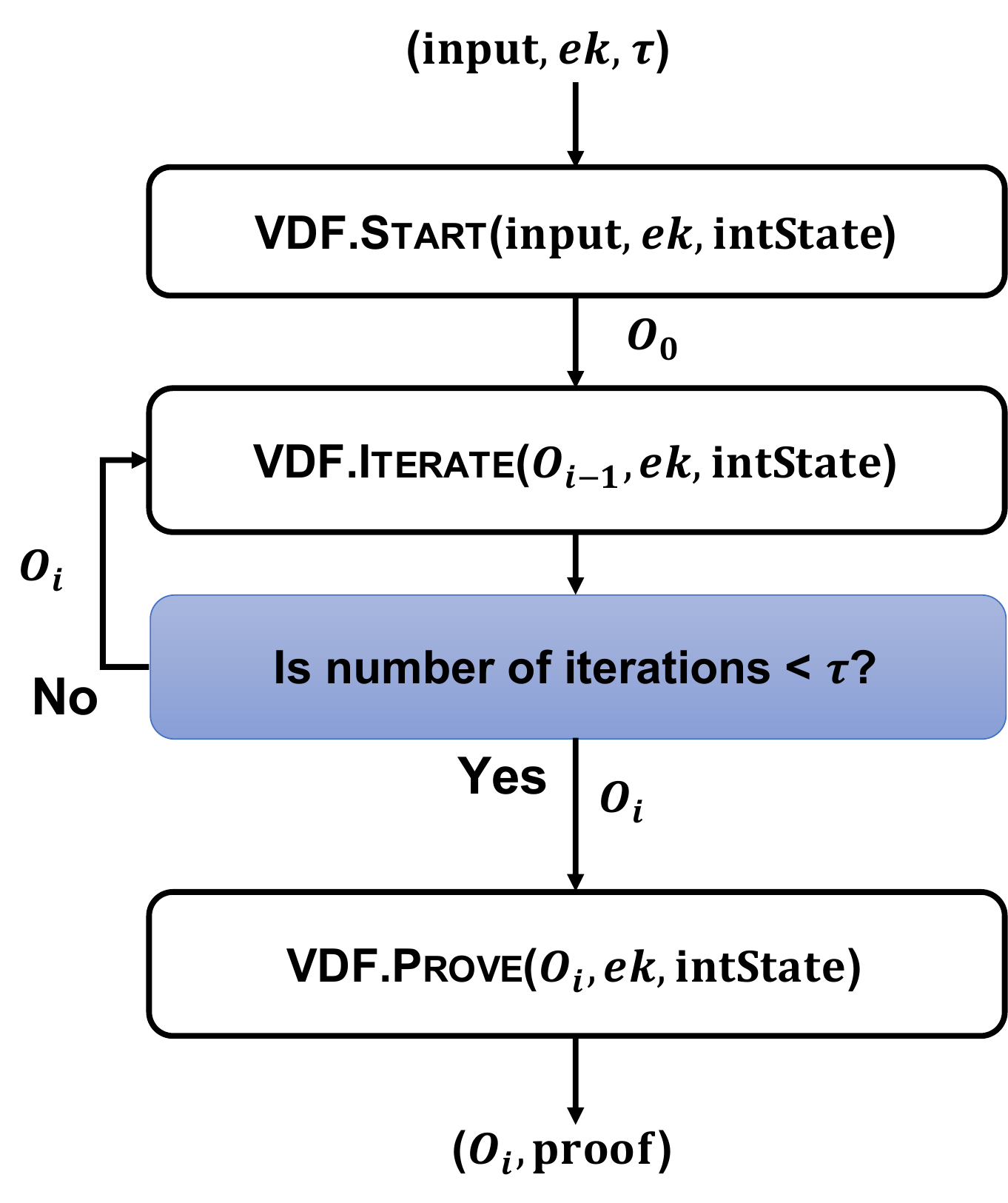}}
~
   \subfloat[ $\rvdf$.\textsc{Eval}($input,ek,s$).]{\label{fig:proposed_VDF}
      \includegraphics[width=.42\textwidth]{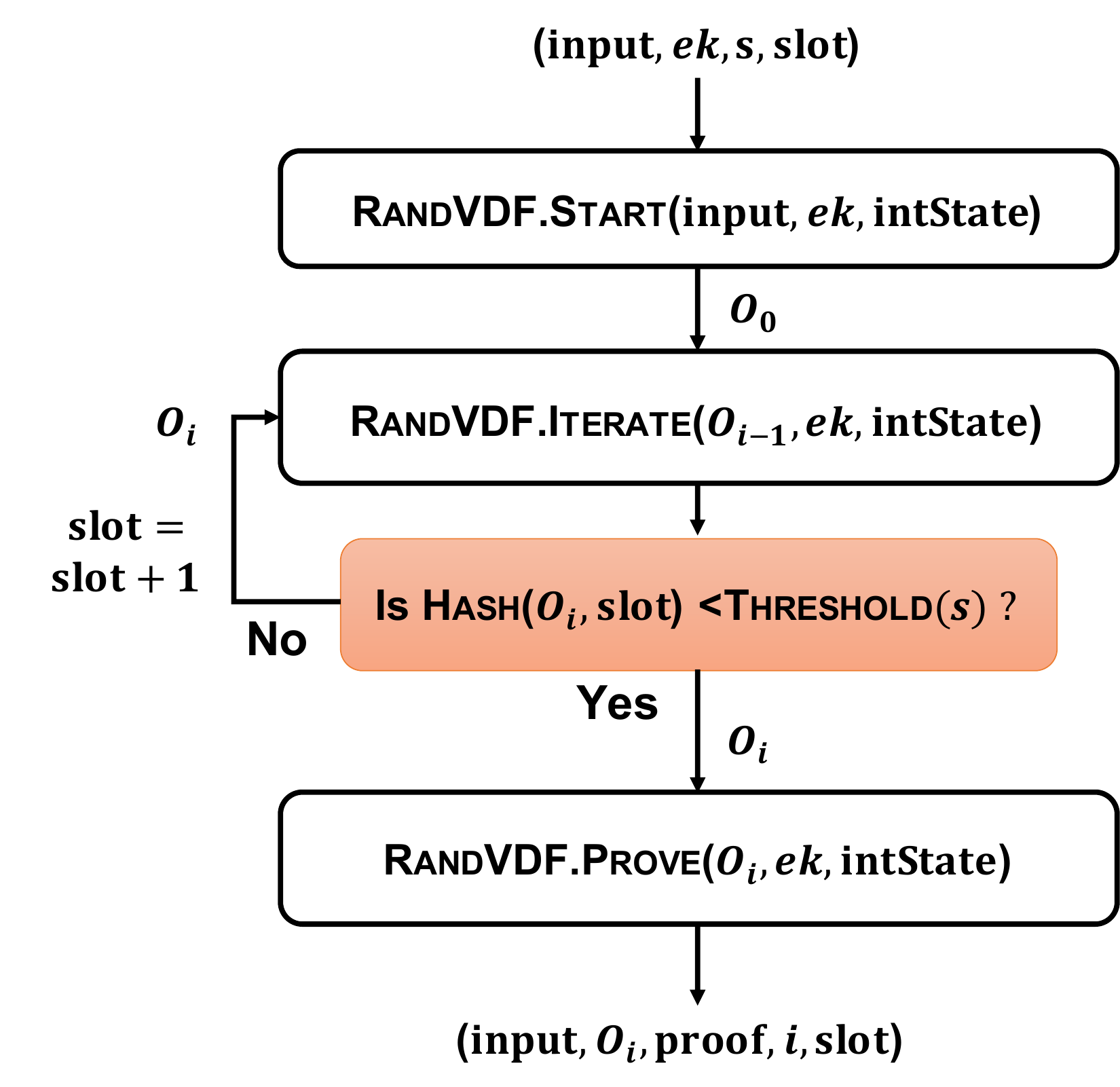}}
   \caption{$\vdf$.{\sc Eval}($\ip,ek,\tau$) requires the number of iterations that $\vdf$.{\sc Iterate} should run. On the other hand, $\rvdf$.{\sc Eval}($\ip,ek,\s,\Time$) requires the expected number of number of iterations $\rvdf$.{\sc Iterate} (denoted by $\s$) must run.}\label{bs1}
\end{figure}

A key feature of VDF is that if the VDF takes $T$ steps, then the prover should be able to complete the proof in time (nearly) proportional to $T$ and the verifier should be able to verify the proof in (poly)-logarithmic time. This makes it feasible for any node that receives a block to quickly verify that the VDF in the header is indeed correctly computed, without expending the same effort that was expended by the prover. We refer the reader interested in a detailed analysis of these complexities  to Section 6.2 in \cite{pietrzak2018simple} for the efficiency calculation or Section 2.3 in \cite{ephraim2020continuous}. 

\subsection{Protocol description}
The pseudocode for the \protocol is given in Algorithm~\ref{alg:PoSAT}.

\begin{algorithm}[H]
	{\fontsize{9pt}{9pt}\selectfont \caption{PoSAT}
	\label{alg:PoSAT}
\begin{algorithmic}[1]
\Procedure{Initialize}{ } \colorcomment{ all variables are global}
\State $\texttt{blkTree} \gets \textsc{Sync}()$ \colorcomment{syncing with peers}
\State  \unCnfTx $  \gets \phi$ \colorcomment{pool of unconfirmed {\sf tx}s}
\State $\parentBk \gets \texttt{blkTree}.\textsc{Tip}()$ \colorcomment{tip of the longest chain in \texttt{blkTree}}
\State $\rs \gets \texttt{None}$ \colorcomment{will be updated at next epoch beginning}
\State $\Time \gets \texttt{None}$ \colorcomment{will be updated at next epoch beginning}
\State \Return False
\EndProcedure
\vspace{1mm}
\Procedure{PosLeaderElection}{coin} 
\State ($\rvdf.ek$, $\rvdf.vk$),($\textsc{Sign}.vk$, $\textsc{Sign}.sk$) $\gets$ \coin.{\sc Keys()}
\State $\stake \gets {\rm coin}.\textsc{Stake}(\textsc{ SearchChainUp}(\parentBk))$ \colorcomment{update the stake}\label{algo:stake}
\State $\s \gets \textsc{UpdateThreshold}(\stake)$ \colorcomment{update the threshold}
\State $\ip \gets \rs$
\State \maincolorcomment{Calling $\rvdfeval$}
\State $(\ip, \op, \pf, \randiter,\Time) \gets \textsc{RandVDF.Eval}(\ip,ek,\s,\Time)$

\State $\rs \gets \op$ \label{algo:randsourceupdate2} \colorcomment{update source of randomness }
\State  $\state \gets \textsc{Hash}(\parentBk)$
\State $\content  \gets \langle \unCnfTx, \coin, \ip, \rs, \pf, \randiter,
\state, \Time \rangle$

\State \Return $\langle \texttt{header}, \content, \textsc{Sign}(\content, $\textsc{Sign}.sk$)  \rangle$

\EndProcedure
\vspace{1mm}

\Procedure{ReceiveMessage}{\texttt{X}} \colorcomment{receives messages from network}
\If {\texttt{X} is a valid {\sf tx}}
    \State $\unCnfTx  \gets  \unCnfTx \cup\; \{\texttt{X}\}$

	\ElsIf{ {\sc IsValidBlock}(\texttt{X})}
		\If { $\parentBk.\textsc{Level}() < \texttt{X}.\textsc{Level}()$ }
            \State \textsc{ChangeMainChain(\texttt{X})} \colorcomment{if the new chain is longer}
            \State $\parentBk \gets \texttt{X}$ \colorcomment{update the parent block to tip of the longest chain} \label{algo:parent-recording}
            \If {$\texttt{X}.\textsc{Level}() \;\%\; c == 0$} 
                \State $\rs \gets \texttt{X}.\content.\rs$ \label{algo:randsourceupdate1}
                
            \Else
                \State $\rs \gets \rs$ \label{algo:randsourceupdate3}
            \EndIf
    		\If{\texttt{participate} == True}                
    		    \State $\rvdf.\textsc{Reset}()$ \colorcomment{reset the $\rvdf$} \label{algo:reset}
    		\EndIf
    		\State \maincolorcomment{Epoch beginning}
    		\If{(\texttt{X}.\textsc{Level}() \% c == 0) \& (\texttt{participate} == False)}
    		    \label{algo:participate1}
                \State $\Time \gets \texttt{X}.\content.\Time$
                \label{algo:participate2}
                \State \texttt{participate} = True
            \EndIf
    	\EndIf

        \EndIf

\EndProcedure
\vspace{1mm}

\Procedure{IsValidBlock}{\texttt{X}} \colorcomment{returns true if a block is valid}
	\If { {\bf not} {\sc IsUnspent}(\texttt{X}.\content.\coin)}
	 \Return False
	\EndIf
	\If{$\textsc{ParentBlk}(\texttt{X}).\content.\Time \geq X.\content.\Time$}     \label{algo:time-ordering}
	    \State \Return False \colorcomment{ensuring time ordering} 
	\EndIf
	\State $\s \gets \textsc{UpdateThreshold}(\textsc{ParentBlk}(\texttt{X}))$
	\If{\textsc{Hash}(\texttt{X}.\content.\{\rs,\Time\}) $>$ \textsc{Threshold}$(\texttt{s})$}
	\Return False
	\EndIf
	\State \maincolorcomment{verifying the work}
	\State \Return $\textsc{RandVDF.Verify}(\texttt{X}.\coin.vk,\texttt{X}. \content.\{\ip,\rs,\pf,\randiter\})$ \label{algo:randvdfverify}
\EndProcedure
\vspace{1mm}

\vspace{1mm}
\Procedure{Main}{ }  \colorcomment{main function}
    \State \texttt{participate} $=$ \textsc{Initialize}()
    \State \textsc{StartThread(ReceiveMessage)} \colorcomment{parallel thread for receiving messages}
    \While{True}
        \If{\texttt{participate} == True}
            \State \texttt{block} =  \textsc{PosLeaderElection}(\coin)
            \State\textsc{SendMessage}(\texttt{block}) 
        \colorcomment{broadcast to the whole network}
        \EndIf
    \EndWhile
\EndProcedure

\end{algorithmic}
}
\end{algorithm}

\subsubsection{Initialization.} An honest coin $n$ on coming online, calls \textsc{Initialize}() where it obtains the current state of the blockchain, \texttt{blkTree}, by synchronizing with the peers via \textsc{Sync}() and initializes global variables. However,
the coin $n$ can start participating in the leader election only after encountering the next epoch beginning, that is, when the depth of the \texttt{blkTree} is a multiple of $c$. This is indicated by setting $\texttt{participate}_n$ to False. Observe that if the coin $n$ is immediately allowed to participate in leader election, then, the coin $n$ would have to initiate $\rvdfeval$ from the $\rs$ contained in the block at the beginning of the current epoch. Due to the sequential computation in  $\rvdf$, the coin $n$ would never be able to participate in the leader elections for proposing blocks at the tip of the blockchain. In parallel, the coin keeps receiving messages and processes them in \textsc{ReceiveMessage}(). On receiving a valid block that indicates epoch beginning,  $\rs_n$, $\Time_n$ and $\texttt{participate}_n$ are updated accordingly (lines~\ref{algo:randsourceupdate1},~\ref{algo:participate1},~\ref{algo:participate2}) for active participation in leader election. 

\subsubsection{Leader election.}
The coin $n$ records the tip of the longest chain of \texttt{blkTree} in $\parentBk_n$ (line~\ref{algo:parent-recording}) and contests leader election for appending block to it. $\rvdf.\textsc{Eval}(\ip_n,\rvdf.ek_n,\texttt{s}_n)$ is used to compute an unpredictable randomness beacon that imparts unpredictability to leader election. The difficulty parameter $\texttt{s}_n$ is set proportional to the current $\stake_n$ of the coin $n$ using $\textsc{UpdateThreshold}(\stake_n)$ and $\rs_n$ is taken as $\ip_n$. $\rvdf.\textsc{Eval}(\ip_n,ek_n,\texttt{s}_n, \Time_n)$ is an iterative function composed of:
\begin{itemize}
    \item $\rvdf.\textsc{Start}(\ip_n,\rvdf.ek_n,\texttt{IntState}_n)$ initializes the iteration by setting initial value of $\op_n$ to be $\ip_n$. Note that $\texttt{IntState}_n$ is the internal state of the $\rvdf$. 
    \item $\rvdf.\textsc{Iterate}(\op_n,\rvdf.ek_n,\texttt{IntState}_n)$ is the iterator function that updates $\op_n$ in each iteration. At the end of each iteration, it is checked whether $\textsc{Hash}(\op_n,\Time_n)$ is less than $\textsc{Threshold}(\texttt{s}_n)$, which is set proportional to $\texttt{s}_n$. 
    If \texttt{No}, $\Time_n$ is incremented by $1$ and current $\op_n$ is taken as input to the next iteration. If \texttt{Yes}, then it means coin $n$ has won the leader election and $\op_n$ is passed as input to $\rvdf.\textsc{Prove}(.)$. Observe that the number of iterations, $\randiter_n$, that would be required to pass this threshold is unpredictable which lends to randomness beacon. Recall that $\Time_n$ is a counter for number of iterations since genesis. In a PoS protocol, it is normally ensured that the timestamps contained in each block of a chain are ordered in ascending order. Here, in $\protocol$, instead we ensure that the $\Time$ in the blocks of a chain are ordered, irrespective of who proposed it. This is referred to as {\em time-ordering}. The reader can refer to Appendix~\ref{sec:block_enumeration_attack} and \ref{sec:long_range_attack} for further details on what attacks can transpire if time-ordering is not ensured.
    The rationale behind setting $\textsc{Threshold}(\texttt{s}_n)$ proportional to $\texttt{s}_n$ is that even if the stake $\texttt{s}_n$ is sybil over multiple coins, the probability of winning leader election in at least one coin remains the same. See Appendix~\ref{sec:sybil_attack} for detailed discussion.
    \item $\rvdf.\textsc{Prove}(\op_n,\rvdf.ek_n,\texttt{IntState}_n)$ operates on $\op_n$ using  $\rvdf.ek_n$ and $\texttt{IntState}_n$ to generate $\pf_n$ that certifies the iterative computation done in the previous step.
\end{itemize}
\noindent The source of randomness $\rs_n$ can be updated in two ways:
\begin{itemize}
    \item a block, proposed by another coin, at epoch beginning is received  (line~\ref{algo:randsourceupdate1}) 
    \item if coin $n$ wins a leader election and proposes its own block (line~\ref{algo:randsourceupdate2}).
\end{itemize}

While computing $\rvdfeval(.)$, if a block is received that updates $\parentBk_n$, then, $\rvdf.\textsc{Reset}()$ (line~\ref{algo:reset}) pauses the ongoing computation, updates $\s_n$ and continues the computation with updated $\textsc{Threshold}(\texttt{s}_n)$. If $\rs_n$ is also updated, then, $\rvdf.\textsc{Reset}()$ stops the ongoing computation of $\rvdfeval(.)$ and calls \textsc{PoSLeaderElection()}.

\subsubsection{Content of the block.}
Once a coin is elected as a leader, all unconfirmed transactions in its buffer are added to the $\content$. The $\content$ also includes the identity $\coin_n$, $\ip_n$, $\rs_n$, $\pf_n$, $\randiter_n$, $\Time_n$ from $\rvdfeval(.)$. The $\state$ variable in the content contains the hash of parent block, which ensures that the content of the parent block cannot be altered. Finally, the header and the content is signed with the secure signature $\textsc{SIGN}.sk_n$ and the block is proposed. When the block is received by other coins, they check that the time-ordering is maintained (line~\ref{algo:time-ordering}) and verify the work done by the coin $n$ using $\rvdfverify(.)$ (line~\ref{algo:randvdfverify}). Note that the leader election is independent of the content of the block and content of previous blocks. This follows a standard practice in existing PoS protocols such as \cite{badertscher2018ouroboros} and \cite{sleepy} for ensuring that a grinding attack based on enumerating the transactions won't be possible. The reader is referred to Appendix~\ref{sec:content_grinding_attack} for further details. However, this allows the adversary to create multiple blocks with the same header but different content. Such copies of a block with the same header but different contents are known as a ``forkable string” in \cite{kiayias2017ouroboros}. We show in the section~\ref{sec:analysis} that the \protocol is secure against all such variations of attacks.

\subsubsection{Confirmation rule.} A block is confirmed if the block is $k-$deep from the tip of the longest chain. The value of $k$ is determined by the security parameter.

\section{Model}
\label{sec:model}

We will adopt a continuous-time  model. Like the $\Delta$-synchronous model in \cite{pss16}, we assume there is a bounded communication delay $\Delta$ seconds between the honest nodes (the particular value of latency of any transmission inside this bound is chosen by the adversary). 

The blockchain is run on a network of $N$ honest nodes and a set of adversary nodes.  Each node holds a certain number of  coins (proportional to their stake). We allow  nodes to join and leave the network, thus the amount of honest/adversarial stake  which is participating in the protocol varies as a function of time. Recall that, as described in section~\ref{sec:protocols}, a coin coming online can only participate in the leader election after encountering the next epoch beginning.  Let $\lambda_h(t)$ be defined as the stake of the honest coins that are online at time $t$ and has encountered at least one epoch beginning. Thus, $\lambda_h(t)$ is the rate at which honest nodes win leader elections. Let $\lambda_a(t)$ be the stake controlled by the adversary.  We will assume there exist constants $\lambda_{\min}, \lambda_{\max} > 0$  such that 
\begin{align}
    \lambda_{\min} \leq \lambda_h(t) \leq \lambda_{\max} \quad \forall \quad t\geq 0.
\end{align}
The existence of $\lambda_{\max}$ is obvious since we are in a proof-of-stake system, and $\lambda_{\max}$ denotes the rate at which the leader elections are being won if every single stakeholder is online. We need to assume a minimum $\lambda_h(t)$ in order to guarantee that within a bounded time, a new block is created.


 An honest node will construct and publicly reveal the block immediately after it has won the corresponding leader election. However, an adversary can choose to not do so. By ``private block", we  refer to a block whose corresponding computation of $\rvdfeval$ was completed by the adversary earlier than when the block was made public. Also, by ``honest block proposed at time $t$", we mean that the computation of $\rvdfeval$ was completed at time $t$ and then the associated honest block was instantaneously constructed and publicly revealed. 



The evolution of the blockchain can be modeled as a process $\{(\T(t), \C(t), \T^{(p)}(t), \C^{(p)}(t)): t \ge 0, 1 \leq p \leq N\}$, $N$ being the number of honest nodes, where:
\begin{itemize}
    \item $\T(t)$ is a tree, and is interpreted as the {\em mother tree} consisting of all the blocks that are proposed by both the honest and the adversary nodes up until time $t$ (including private blocks at the adversary). 
    \item $\T^{(p)}(t)$ is an induced (public) sub-tree of the mother tree $\T(t)$ in the view of the $p$-th honest node at time $t$. 
    \item $\C^{(p)}(t)$ is the longest chain in the tree $\T^{(p)}(t)$, and is interpreted as the longest chain in the local view of the $p$-th honest node.
   \item $\C(t)$ is the common prefix of all the local chains $\C^{(p)}(t)$ for $1 \leq p \leq N$. 
\end{itemize}

The process evolution is  as follows.

\begin{itemize}

    \item {\bf M0}: $\T(0)= \T^{(p)}(0) = \C^{(p)}(0), 1\leq p\leq N$ is a single root block (genesis).

    
    \item {\bf M1}: There is an independent leader election at every epoch beginning, i.e., at every block in the blocktree at level $c,2c,..., \ell c,...$. The leader elections are won by the adversary according to  independent Poisson processes of rate $\lambda_a(t)$ at time $t$, one for every block at the aforementioned levels.  The adversary can use the leader election won at a block at level $\ell c$ at time $t$ to  propose a block at every block in the next $c-1$ levels $\ell c, \ell c +1, ..., \ell c + c - 1$ that are present in the tree $\T(t)$. We refer the reader to Figure~\ref{fig:models} for a visual representation. 
    

    
    \item {\bf M2}: Honest blocks are proposed at a total rate of $\lambda_h(t)$ at time $t$ across all the honest nodes at the tip of the chain held by the mining node $p$,  $\C^{(p)}(t)$. 
    
    \item {\bf M3}:  The adversary can replace $\T^{(p)}(t^-)$ by another sub-tree $\T^{(p)}(t)$ from $\T(t)$ as long as the new sub-tree $\T^{(p)}(t)$ is an induced sub-tree of the new tree $\T^{(p)}(t)$, and can update $\C^{(p)}(t^-)$ to a longest chain in $T^{(p)}(t)$. \footnote{All jump processes are assumed to be right-continuous with left limits, so that $\C(t),\T(t)$ etc. include the new arrival if there is a new arrival at time $t$.} 
    
    \end{itemize}

We highlight the capabilities of the adversary in this model:

\begin{itemize}
\item {\bf A1}: Can choose to propose block on multiple blocks of the tree $\T(t)$ at any time.
\item {\bf A2}: Can delay the communication of blocks between the honest nodes, but no more than $\Delta$ time.
\item {\bf A3}: Can broadcast private blocks at times of its own choosing: when private blocks are made public at time $t$ to node $p$, then these blocks are added to $\T^{(p)}(t^-)$ to obtain $\T^{(p)}(t)$. Note that, under $\Delta$-synchronous model, when private blocks appear in the view of some honest node $p$, they will also appear in the view of all other honest nodes  by time $t+\Delta$. 
\item {\bf A4}: Can switch the chain where the $p$-th honest node is proposing block, from one longest chain to another of equal length, even when its view of the tree does not change, i.e., $\T^{(p)}(t) = \T^{(p)}(t^-)$ but $\C^{(p)}(t) \not = \C^{(p)}(t^-)$. 
\end{itemize}
It is to be noted that we don't consider the adversary to be {\em adaptive} in the sense that, although adversarial and honest nodes can join or leave the system as they wish, an adversary can never turn honest nodes adversarial. In order to defend against an adaptive adversary, key evolving signature schemes can be used \cite{david2018ouroboros}. However, in order to keep the system simple, we don't consider adaptive adversary.


\begin{figure*}[h!]
    \centering
    \includegraphics[width=\textwidth]{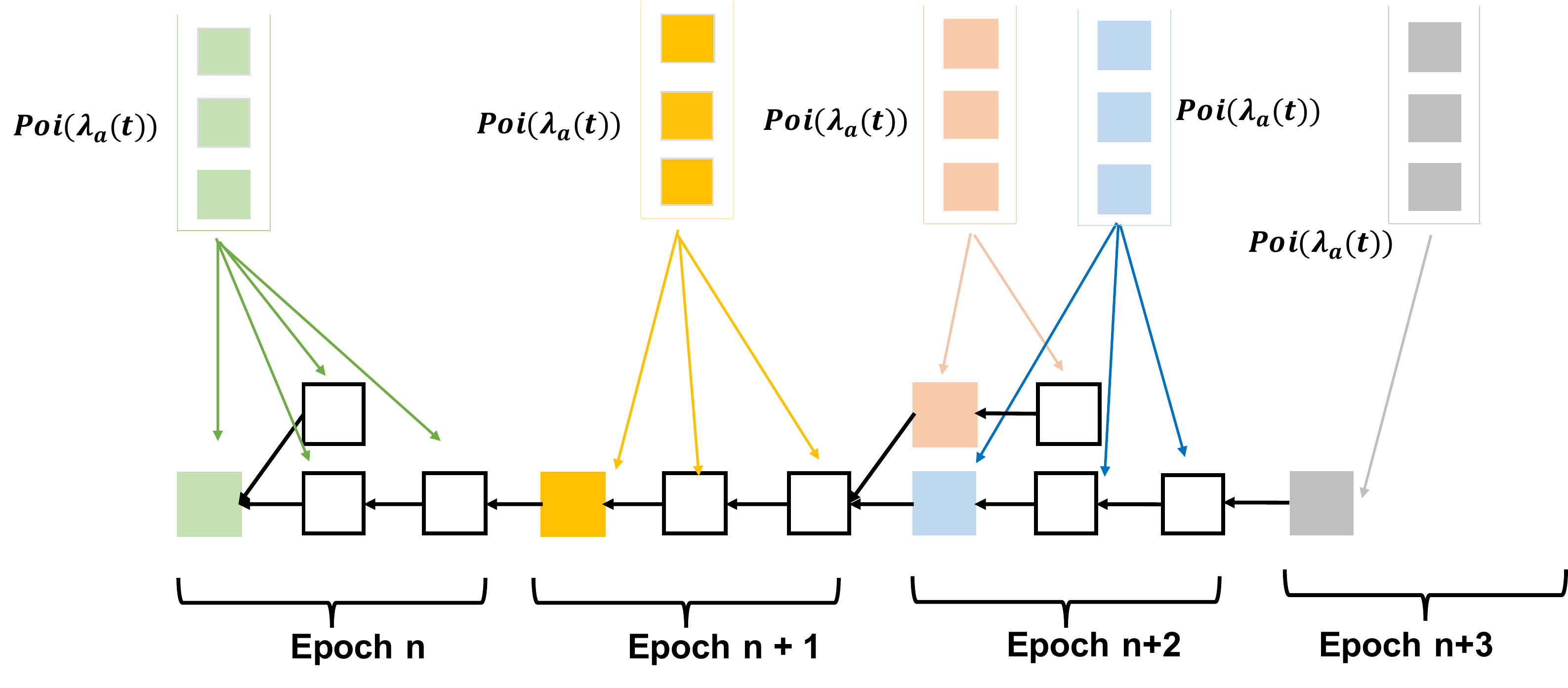}
    \caption{There is a separate randomness generated for every block in the modulo $c$ position. Blocks generated from that randomness at time  $t$ can attach to any block inside the next $c-1$ blocks that are present in the tree $\T(t)$.}
    \label{fig:models}
\end{figure*}

Proving the security (persistence and liveness) of the protocol boils down to  providing a guarantee that the chain $\C(t)$ converges fast as $t \rightarrow \infty$ and that honest blocks enter regularly into $\C(t)$ regardless of the adversary's strategy.

\section{Security Analysis}
\label{sec:analysis}

Our goal is to generate a transaction ledger that satisfies persistence and liveness as defined in \cite{backbone}. Together, persistence and liveness guarantee robust transaction ledger; honest transactions will
be adopted to the ledger and be immutable.
\begin{definition}[from \cite{backbone}]
    \label{def:public_ledger}
    A protocol $\Pi$ maintains a robust public transaction ledger if it organizes the ledger as a blockchain of transactions and it satisfies the following two properties:
    \begin{itemize}
        \item (Persistence) Parameterized by $\tau \in \mathbb{R}$, if at a certain time a transaction {\sf tx} appears in a block which is mined more than $\tau$ time away from the mining time of the tip of the main chain of an honest node (such transaction will be called confirmed), then {\sf tx} will be confirmed by
        all honest nodes in the same position in the ledger.
        \item (Liveness) Parameterized by $u \in \mathbb{R}$, if a transaction {\sf tx} is received by all honest nodes for more than time $u$, then all honest nodes will contain {\sf tx} in the same place in the ledger forever.
    \end{itemize}
\end{definition}

\subsection{Main security result}

To state our main security result, we need to define some basic notations.

Recall that, as described in section~\ref{sec:protocols}, a coin coming online can only participate in the leader election after encountering the next epoch beginning. This incurs a random waiting delay for the coin before it can actively participate in the evolution of the blockchain. Hence, the honest mining rate $\lambda_h(t)$, defined in Section \ref{sec:model}  as the stake of the honest coins that are online at time $t$ and has encountered at least one epoch beginning,  is a (random) process that depends on the dynamics of the blockchain. Hence, we cannot state a security result based on conditions on $\lambda_h(t)$. Instead, let us define 
$\lambda_h^c(t)$ as the stake of the honest coins that are online at time $t$ and has been online since at least time $t-\sigma(c)$, where
\begin{equation}
    \sigma(c) = (c-1)\left(\Delta + \frac{1+\kappa}{\lambda_{\min}}\right).
\end{equation}
Here, $\kappa$ is the security parameter. Intuitively, $\sigma(c)$ is a high-probability worst-case waiting delay, in seconds, of a coin for the next epoch beginning. Note that $\lambda^c_h(t)$ depends only on the stake arrival process and not on the blockchain dynamics.


The theorem below shows that the the private attack threshold
yields the true security threshold:

\begin{theorem}
\label{thm:main}
If 
\begin{align}
\label{eq:cond}
    \frac{\lambda^c_h(t)}{1+\lambda_{\max}\Delta} > \phi_c\lambda_a(t)  \qquad  \mbox{for all $t>0$},
\end{align}
then the \protocol generate transaction ledgers such that each transaction tx satisfies persistence (parameterized by $\tau = \rho$)  and liveness (parameterized by $u=\rho$) in Definition~\ref{def:public_ledger}  with probability at least $1 - e^{-\Omega(\min\{\rho^{1-\epsilon},\kappa\})}$, for any $\epsilon > 0$. The constant $\phi_c$ is defined in (\ref{eq:phi_c}), with $\phi_1 = e$ and $\phi_c \rightarrow 1$ as $c \rightarrow \infty$.
\end{theorem}

In order to prove Theorem~\ref{thm:main}, we utilize the concept of blocktree partitioning and Nakamoto blocks that were introduced in \cite{dembo2020everything}. We provide a brief overview of these concepts here.

Let $\tau^h_i$ and $\tau^a_i$ be the time when the $i$-th honest and adversary blocks are proposed, respectively; $\tau^h_0 = 0$ is the  time when the genesis block is proposed, which we consider as the $0$-th honest block. 

\begin{defn}
{\bf Blocktree partitioning}
Given the mother tree $\T(t)$, define for the $i$-th honest block $b_i$, the {\em adversary tree} $\T_i(t)$ to be the sub-tree of the mother tree $\T(t)$ rooted at $b_i$  and consists of all the adversary blocks that can be reached from $b_i$ without going through another honest block. The mother tree $\T(t)$ is partitioned into sub-trees $\T_0(t),\T_1(t), \ldots \T_j(t)$, where the $j$-th honest block is the last honest block that was proposed before time $t$.
\end{defn}
The sub-tree $\T_i(t)$ is born at time $\tau^h_i$ as a single block $b_i$ and then grows each time an adversary block is appended to a chain of adversary blocks from $b_i$. Let $D_i(t)$ denote the depth of $\T_i(t)$; $D_i(\tau_i^h) = 0$.  

\begin{defn} \cite{ren}
The $j$-th honest block proposed at time $\tau^h_j$ is called a {\em loner} if there are no other honest blocks proposed  in the time interval $[\tau^h_j - \Delta, \tau^h_j + \Delta]$.
\end{defn}

\begin{defn}
\label{def:fictitious}
Given honest block proposal times $\tau^h_i$'s, define a honest fictitious tree $\T_h(t)$ as a tree which evolves as follows:
\begin{enumerate}
    \item $\T_h(0)$ is the genesis block.
    \item The first honest block to be proposed and all honest blocks within $\Delta$ are all appended to the genesis block at their respective proposal times to form the first level. 
    \item The next honest block to be proposed and all honest blocks proposed within time $\Delta$ of that are added to form the second level (which first level blocks are parents to which new blocks is immaterial) . 
    \item The process repeats.
\end{enumerate}
Let $D_h(t)$ be the depth of $\T_h(t)$.
\end{defn}

\begin{defn}
\label{defn:nak_blk_gen}
({\bf Nakamoto block}) Let us define:
\begin{equation}
\label{eqn:Et}
    E_{ij} = \mbox{event that $D_i(t)  < D_h(t-\Delta) - D_h(\tau^h_i+\Delta)$ for all $t > \tau^h_j + \Delta$}.
\end{equation}
The $j$-th honest block is called a {\em Nakamoto block} if it is a loner and
\begin{equation}
F_j = \bigcap_{i = 0}^{j-1} E_{ij}
\end{equation}
occurs.
\end{defn}
See  Figure 5 in \cite{dembo2020everything} for illustration of the concepts of blocktree partitioning and Nakamoto blocks. 
\begin{lemma} (Theorem 3.2 in \cite{dembo2020everything})
\label{lem:nak_blk}({\bf Nakamoto blocks stabilize})
If the $j$-th honest block is a Nakamoto block, then it will be in the longest chain $\C(t)$ for all $t > \tau^h_j+\Delta$. 
\end{lemma}

Lemma~\ref{lem:nak_blk} states that Nakamoto
blocks remain in the longest chain forever. The question is whether
they exist and appear frequently regardless of the adversary strategy. If they do, then the protocol has liveness and persistence:
honest transactions can enter the ledger frequently through the
Nakamoto blocks, and once they enter, they remain at a fixed location in the ledger. More formally, we have the following result.

\begin{lemma} (Lemma 4.4 in \cite{dembo2020everything})
\label{lem:nak_secure}
Define $B_{s,s+t}$ as the event that there is no Nakamoto blocks in the time interval $[s,s+t]$ where $t \sim \Omega\left(\left[\frac{c-1}{\phi_c-1}\right]^2\right)$. If \begin{equation}
    P(B_{s,s+ t}) < q_t < 1
\end{equation}
for some $q_t$ independent of $s$ and the adversary strategy, then the \protocol generates   transaction ledgers such that each transaction tx satisfies {\em persistence} (parameterized by $\tau=\rho$) and {\em liveness} (parameterized by $u=\rho$) in Definition~\ref{def:public_ledger} with probability at least $1-q_\rho$.
\end{lemma}

In order to prove Lemma~\ref{lem:nak_secure}, we proceed in six steps as illustrated in Fig.~\ref{fig:proof_flowchart}. 
\begin{figure}[H]
    \includegraphics[width=0.7\textwidth]{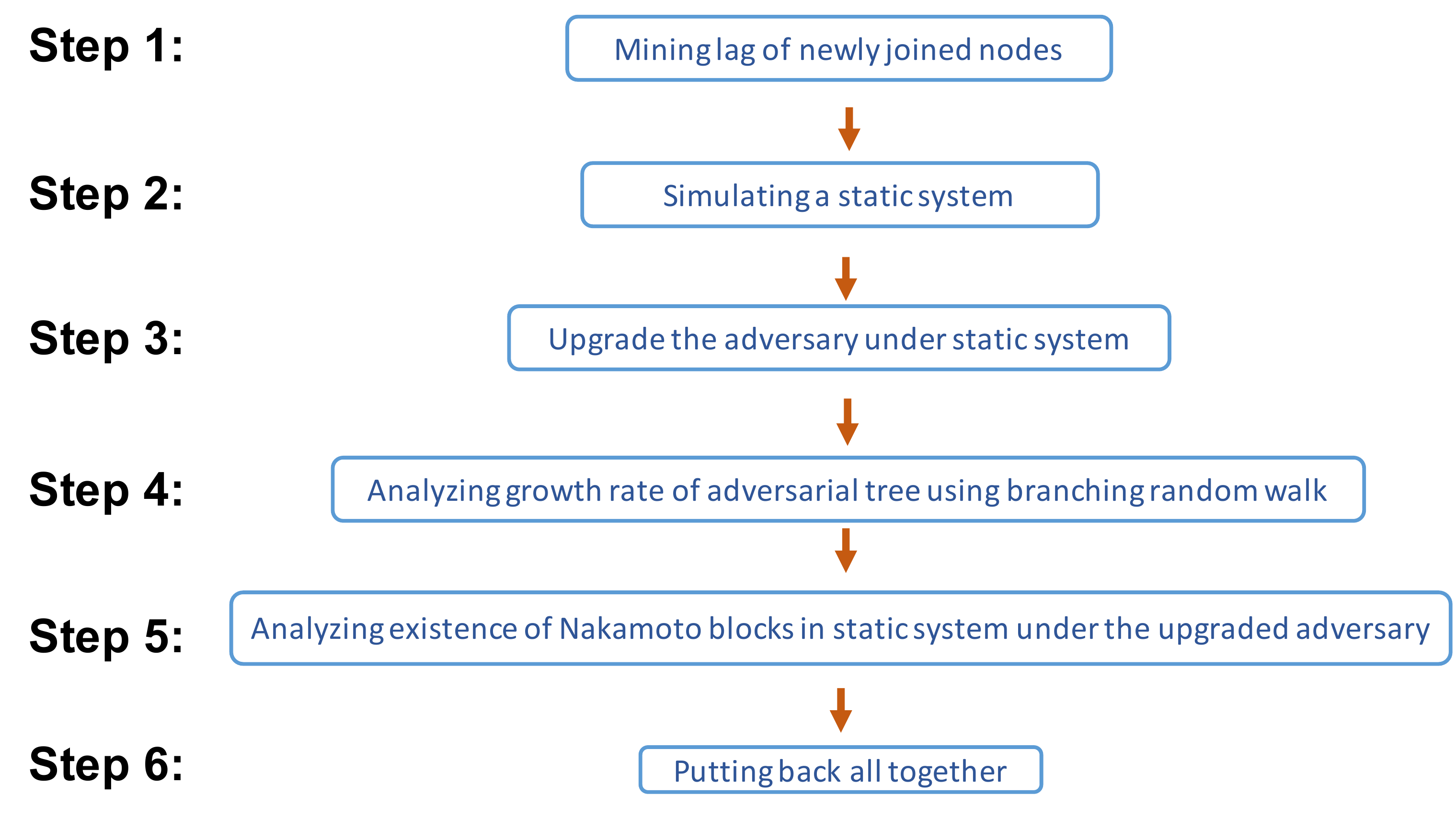}
    \centering
    \caption{Flowchart of the proof for  Lemma~\ref{lem:nak_secure}.}
    \label{fig:proof_flowchart}
\end{figure}

\subsection{Step $1$: Mining lag of newly joined nodes}
From section~\ref{sec:model}, recall that $\lambda_h(t)$ is defined as the stake of the coins that are online at time $t$ but has encountered at least one epoch beginning. That implies, within an epoch, $\lambda_h(t)$ is the effective honest stake that can be used to contribute towards the growth of the longest chain; it remains constant and gets updated only at the epoch beginning. 
In order to analyze the effect of this lag in a honest node to start mining, we simulate a new dynamic available system, $dyn2$, where, at time $t$, an honest coin can contribute towards the growth of the longest chain if it has been online in the original dynamic system since at least time $t-\sigma(c)$, where, $\sigma(c) > 0$. Recall that $\lambda^c_h(t)$ be defined as the stake of the coins that are online at time $t$ in the original dynamic system and has been online since at least $t-\sigma(c)$. Clearly, $\lambda^c_h(t)$ is the rate at which the honest nodes win leader election at time $t$ in $dyn2$. We have the following relationship between the original dynamic available system  and $dyn2$.
\begin{lemma}
\label{lem:step1-lemma}
For the dynamic available system $dyn2$ and for all $s,t > 0$, define $B^{dyn2}_{s,s+t}$ as the event that there are no Nakamoto blocks in the time interval $[s,s+t]$. Let $\kappa_0$ be the solution for the equation $\ln{\left(\frac{\lambda_{\max}}{\lambda_{\min}}(1+\kappa)\right)} = \kappa$. Then, for $\sigma(c) = (c-1)\left(\Delta + \frac{1+\kappa}{\lambda_{\min}}\right)$ and $\kappa >> \kappa_0$, we have 
$$
P(B_{s,s+t}) \leq P(B^{dyn2}_{s,s+t}) +  e^{-\Omega(\kappa)}.
$$
\end{lemma}
\noindent The proof is given in Appendix~\ref{sec:step1-proof}.

\subsection{Step $2$: Simulating a static system}
Without loss of generality, we assume that the adversarial power is boosted such that $\lambda_a(t)$ satisfies (\ref{eq:cond}) with equality for all $t$. Let us define $\eta$  such that $\lambda_a(t) = (1-\eta) \lambda_h(t)$ for all $t$. Let $\lambda_h$ be some positive constant. Taking $dyn2$ as the base, we simulate a static system, $ss0$, where both honest nodes and adversary win leader elections with constant rates $\lambda_h$ and $\lambda_a$ satisfying $\lambda_a = (1-\eta)\lambda_h$. This requires, for a local time $t>0$ in $dyn2$, defining a new local time $\alpha(t)$ for $ss0$ such that
\begin{align}
    \label{eq:time-relations}
    \lambda^c_h(u)du = \lambda_h d\alpha \implies \alpha(t) = \int_{0}^{t} \frac{\lambda^c_h(u)}{\lambda_h}du.
\end{align}
Additionally, for every arrival of an honest or adversarial block in $dyn2$ at a particular level at a tree, there is a corresponding arrival in $ss0$ at the same level in the same tree. For a time $t$ in the local clock of $dyn2$, let $\Delta^{ss0}(t)$ be the network delay of $dyn2$ measured with reference to the local clock of $ss0$. Using (\ref{eq:time-relations}), we have
\begin{align}
\label{eq:delta-inequality}
    \frac{\lambda_{\min}}{\lambda_h}\Delta \leq \Delta^{ss0}(t) \leq \frac{\lambda_{\max}}{\lambda_h}\Delta.
\end{align}
We have the following relationship between $dyn2$ and $ss0$.
\begin{lemma}
    \label{lem:step2-lemma}
    Consider the time interval $[s,s+t]$ in the local clock of $dyn2$. For the static system $ss0$, define $B^{ss0}_{\alpha(s),\alpha(s+t)}$ as the event that there are no Nakamoto blocks in the time interval $[\alpha(s),\alpha(s+t)]$ in the local clock of $ss0$. Then,
    $$P(B^{dyn2}_{s,s+t}) = P(B^{ss0}_{\alpha(s),\alpha(s+t)}).$$
\end{lemma}
\noindent  The proof for this lemma is given in Appendix~\ref{sec:step2-proof}.

\subsection{Step $3$: Upgrading the adversary}
\label{sec:step-3}
As the occurrence of Nakamoto blocks is a race between the fictitious honest tree and the adversarial trees from the previous honest blocks, we next turn to an analysis of the growth rate of an adversary tree. However, the growth rate of an adversarial tree would now depend on the location of the root honest block within an epoch which adds to the complexity of the analysis. To get around this complexity, we simulate a new static system, $ss1$ in which the adversary, on winning a leader election after evaluating $\rvdfeval$ and appending a block to an honest block (that is, growing a new adversarial tree), is given a gift of chain of $c-1$ extra blocks for which the adversary doesn't have to compute $\rvdfeval$. Thus, the adversary has to compute only one $\rvdfeval$ for the chain of first $c$ blocks in the adversarial tree. At this point, the adversary can assume a new epoch beginning and accordingly update $\rs$. Hereafter, the evolution of $\rs$ follows the rules in $ss0$. Note that the local clock for both the static systems $ss0$ and $ss1$ are same. Now, we have the following relationship between $ss0$ and $ss1$.
\begin{lemma}
    \label{lem:step3-lemma}
    Consider the time interval $[s,s+t]$ in the local clock of $dyn2$. For the static system $ss1$, define $B^{ss1}_{\alpha(s),\alpha(s+t)}$ as the event that there are no Nakamoto blocks in the time interval $[\alpha(s),\alpha(s+t)]$ in the local clock of $ss1$. Then,
    $$P(B^{ss0}_{\alpha(s),\alpha(s+t)}) \leq P(B^{ss1}_{\alpha(s),\alpha(s+t)}).$$
\end{lemma}
\noindent The proof for this lemma is given in Appendix~\ref{sec:step3-proof}.

For analyzing $P(B^{ss1}_{\alpha(s),\alpha(s+t)})$, we first consider an arbitrary static system $ss2$ where both honest nodes and adversary win leader elections with constant rates $\lambda_h$ and $\lambda_a$, respectively, the honest nodes follows $\protocol$, the adversary has similar additional power of gift of chain of $c-1$ blocks as in $ss1$ but the network delay is a constant, say $\Delta'$. For some $s',t' > 0$ in the local clock of the static system $ss2$, we will  determine an upper bound on $P(B_{s',s'+t'}^{ss2})$ in Sections~\ref{sec:adversarial-tree-growth} - \ref{sec:Nakamoto-block-existence} and then use this result to obtain an upper bound on $P(B^{ss1}_{\alpha(s),\alpha(s+t)})$ in Section~\ref{sec:putting-all-together}.

\subsection{Step $4$: Growth rate of the adversarial tree}
\label{sec:adversarial-tree-growth}
For time $t'>0$, let $\hat{\mathcal{T}}_i(t')$ represents the adversarial tree in $ss2$ with $i^{th}$ honest block as its root. The depth $D_i(t')$ at time $t'$ in the local clock of $ss2$ is defined as the maximum depth of the blocks of $\hat{\mathcal{T}}_i(t')$ at time $t'$. In Lemma~\ref{lem:step4-lemma}, we evaluate the tail bound on  $D_i(t')$. 
\begin{lemma}
  \label{lem:step4-lemma}
  For $x>0$ so that $\eta_c\lambda_a t'+x$ is an integer,
  \begin{equation}
    \label{eq:c-tail}
    P(D_i(t')\geq \phi_c\lambda_a t'+cx) \leq  e^{-\theta_c^* t'} e^{(\eta_c\lambda_a t'+x-1)\Lambda_c(\theta_c^*)}g(t') .
  \end{equation}
  where $\phi_c = c\eta_c$, $g(t') = \sum_{i_1 \ge 1} \int_{0}^{t'} \frac{\lambda_a^{i_1}u^{i_1 - 1}e^{-\lambda_a u}}{\Gamma(i_1)} e^{\theta_c^* u} du$, $\Lambda_c(\theta_c) = \log(-\lambda_a^c/\theta_c(\lambda_a-\theta_c)^{c-1})$ and $\theta_c^*$ is the solution for the equation $\Lambda_c(\theta) = \theta {\dot \Lambda_c(\theta)} $
\end{lemma}
\noindent Details on the analysis of $\hat{\mathcal{T}}_i(t')$ and the proof of Lemma~\ref{lem:step4-lemma} are
in Appendix~\ref{sec:growth_rate}.

\subsection{Step $5$: Existence of Nakamoto blocks}
\label{sec:Nakamoto-block-existence}
With Lemma~\ref{lem:step4-lemma}, we show below that in the static system $ss2$ in the regime $\phi_c \lambda_a <  \frac{\lambda_h}{1 + \lambda_{h}\Delta'}$, Nakamoto blocks has a non-zero probability of occurrence. 
\begin{lemma}
\label{lem:step5-lemma-1}
If
$$\phi_c \lambda_a <  \frac{\lambda_h}{1 + \lambda_{h}\Delta'},$$
then, in the static system $ss2$, there is a $p > 0$ such that the probability of the $j-$th honest block being a Nakamoto block is greater than $p$ for all $j$.
\end{lemma}
The proof of this result can be found in Appendix~\ref{sec:nakamoto_block}. 

Having established the fact that Nakamoto blocks occurs with
non-zero frequency, we can bootstrap on Lemma~\ref{lem:step5-lemma-1} to get a bound
on the probability that in a time interval $[s',s' + t']$, there are no
Nakamoto blocks, i.e. a bound on $P(B_{s',s'+t'})$.
\begin{lemma}
\label{lem:step5-lemma-2}
If 
$$\phi_c \lambda_a <  \frac{\lambda_h}{1 + \lambda_{h}\Delta'},$$
then for any $\epsilon > 0$,
there exist constants $\bar a_\epsilon,\bar A_\epsilon$ so that for all $s'\geq 0$ and $t' > \max\left\{\left(\frac{2\lambda_h}{1-\eta}\right)^2\left(\frac{c-1}{\phi_c-1}\right)^2, \left[(c-1)\left(\Delta' + \frac{1}{\lambda_{\min}}\right)\right]^2\right\}$, we have
\begin{equation}
\label{eqn:qst_strong}
P(B^{ss2}_{s',s'+t'}) \leq \bar A_\epsilon \exp(-\bar a_\epsilon t'^{1-\epsilon})
\end{equation}
where $\bar a_\epsilon$ is a function of $\Delta'$.
\end{lemma}
The proof of this result can be found in Appendix~\ref{sec:B_s_s_plus_t}.

\subsection{Step $6$: Putting back all together}
\label{sec:putting-all-together}
In this section, we use the results from Section~\ref{sec:Nakamoto-block-existence} to upper bound $P(B^{ss1}_{\alpha(s),\alpha(s+t)})$ and hence, $P(B_{s,s+t})$. 

Using equation~\ref{eq:time-relations}, we have $\phi_c \lambda_a(t) <  \frac{\lambda^c_h(t)}{1 + \lambda_{\max}\Delta} \iff \phi_c \lambda_a <  \frac{\lambda_h}{1 + \lambda_{\max}\Delta}$. Then, we have the following lemma:
\begin{lemma}
\label{lem:step6-lemma}
If 
$$\phi_c \lambda_a(t) <  \frac{\lambda^c_h(t)}{1 + \lambda_{\max}\Delta},$$
then for any $\epsilon > 0$
there exist constants $\bar a_\epsilon,\bar A_\epsilon$ so that for all $s\geq 0$ and  $t > \max\left\{\left(\frac{2\lambda_h}{1-\eta}\right)^2\left(\frac{\lambda_h}{\lambda_{\min}}\right)\left(\frac{c-1}{\phi_c-1}\right)^2, \left(\frac{\lambda_h}{\lambda_{\min}}\right)\left[(c-1)\left(\Delta + \frac{1}{\lambda_{\min}}\right)\right]^2\right\}$, we have
\begin{equation}
P(B_{s,s+t}) \leq \bar A_\epsilon \exp(-\bar a_\epsilon t^{1-\epsilon}) + e^{-\Omega(\kappa)}.
\end{equation}
\end{lemma}
The proof for this result is given in Appendix~\ref{sec:step6-lemma-proof}. Then, combining Lemma~\ref{lem:step6-lemma} with Lemma~\ref{lem:nak_secure} implies Theorem~\ref{thm:main}.

\section{Discussion} \label{sec:discuss}
In this section, we discuss some of the practical considerations in adopting \protocol. 


A key question in \protocol is what is the right choice of $c$? If $c$ is low, say $10$, then the security threshold is approximately $1.58$. At $c=10$, the protocol is fully unpredictable and the confirmation latency is not too high. Also, any newly joining honest node has to wait for around $10$ inter-block arrivals  before it can participate in leader election. Thus, if there is a block arrival every second, then, the node has to wait for $10$ secs. In any standard blockchain, there is always a bootstrap period for the node to ensure that the state is synchronized with the existing peers and $10$ secs is negligible as compared to the bootstrap period.

In \protocol, a separate $\rvdf$ needs to be run for  each public-key. In a purely decentralized implementation, all nodes may not have the same rate of computing VDF. This may disadvantage nodes whose rate of doing sequential computation is slower. One approach to solve this problem is to build open-source hardware for VDF - this is already under way through the VDF Alliance. Even under such a circumstnace, it is to be expected that nodes that can operate their hardware in idealized circumstances (for example, using specialized cooling equipment) can gain an advantage. A desirable feature of our protocol is that gains obtained by a slight advantage in the VDF computation rate are bounded. For \protocol, a combination of the VDF computation rate and the stake together yields the net power weilded by a node, and as long as a majority of such power is controlled by honest nodes, we can expect the protocol to be safe. 

In our \protocol specification, the difficulty parameter for the computation of $\rvdfeval$ was assumed to be fixed. This threshold was chosen based on the entire stake being online - this was to ensure that forking even when all nodes are present remains small, i.e., $\lambda_{\max} \Delta$ remained small. In periods when far fewer nodes are online, this leads to a slowdown in confirmation latency. A natural way to mitigate this problem is to use a variable mining threshold based on past history, similar to the adaptation inherent in Bitcoin. A formal analysis of Bitcoin with variable difficulty was carried out in \cite{nakamoto_bounded_delay,backbone_var_difficulty}, we leave a similar analysis of our protocol for future work. 

In our protocol statement, we have used the $\rvdf$ directly on the  randomness prevRand and the public key. The $\rvdf$ ensures that any other node can only predict a given node's leadership slot at the instant that it actually wins the VDF lottery. However, this still enables an adversary to predict the leadership slots of nodes that are offline and can potentially bribe them to come online to favor the adversary. In order to eliminate this exposure, we can replace the hash in the mining condition by using a verifiable random function \cite{vrf,vrf2} (which is calculated using the node's secret key but can be checked using the public key). This ensures that an adversary which is aware of all the public state as well as private state of all {\em online} nodes (including their VRF outputs) still cannot predict the leadership slot of any node ahead of the time at which they can mine the block. This is because, such an adversary does not have access to the VRF output of the offline nodes.

There are two types of PoS protocols: one favoring liveness under dynamic availability and other favoring safety under asynchrony. BFT protocols fall into the latter class and lack dynamic availability. One shortcoming of the longest chain protocol considered in the paper is the reduced throughput and latency compared to the fundamental limits; this problem is inherited from the Nakamoto consensus for PoW \cite{bitcoin}. However, a recent set of papers address these problems in PoW (refer Prism \cite{bagaria2019proof}, OHIE \cite{yu2018ohie} and Ledger Combiners \cite{fitzi2020ledger}). Adaptations of these ideas to the PoSAT protocol is left for future work. Furthermore, our protocol, like Nakamoto, does not achieve optimal chain quality. Adopting ideas from PoW protocols with optimal chain quality, such as Fruitchains \cite{pass2017fruitchains}, is also left for future work.

Finally, while we specified \protocol in the context of proof-of-stake, the ideas can apply to other mining modalities - the most natural example is proof-of-space. We note that existing proof-of-space protocols like Chia \cite{cohen2019chia}, use a VDF for a fixed time, thus making the proof-of-space challenge predictable. In proof-of-space, if the predictability window is large, it is possible to use slow-storage mechanisms such as magnetic disks (which are asymmetrically available with large corporations) to answer the proof-of-space challenges. Our solution of using a RandVDF can be naturally adapted to this setting, yielding unpredictability as well as full dynamic availability.

\section*{Acknowledgement}

DT wants to thank Ling Ren for earlier discussions on dynamic availability of Proof-of-Stake protocols. 






\bibliographystyle{acm}
\bibliography{posPrism}

\newpage
\appendix
\section*{Appendix}


\section{Suite of Possible Attacks Under Dynamic Availability}
In this section, we describe the suite of possible attacks under dynamic availability in PoS systems. This attacks are possible even under static stake. We also discuss some design recommendations for mitigating against such attacks in PoS systems.

\subsection{Content-grinding attack}
\label{sec:content_grinding_attack}
Referring to Fig~\ref{fig:proposed_VDF}, we note that the content of the block, namely the transactions, were not used in determining whether the $\textsc{Threshold(s)}$ is satisfied or not. If we instead checked whether $\textsc{Hash}(O_i,\Time, \texttt{transactionList}) < \textsc{Threshold(s)} $ instead of  $\textsc{Hash}(O_i,\Time) < \textsc{Threshold(s)} $ the protocol loses security due to the ability of adversary to choose the set of transactions in order to increase its likelihood of winning the leadership certificate.

In such a case, the adversary can get 
unlimited advantage by performing 
such content-grinding by parallel computation over different sets and orders of transactions. 
We note that in PoW the adversary does not gain any advantage by performing such content grinding, since it is equivalent to grinding on the Nonce, which is the expected behavior anyway.

\subsection{Sybil attack}
\label{sec:sybil_attack}
One natural attack in PoS for an adversary to sybil the stake contained in a single coin and distribute it across multiple coins which might increase the probability of winning a leader election from at least one of the coins. We describe next that having difficulty parameter in $\rvdfeval$ proportional to the stake of the coin defends against such sybil attack. 
Let us consider $\textsc{H}$ to be the value of the hash function on the output of VDF in an iteration, $R$ to be range of this hash function and $th$ be the difficulty parameter that is proportional to the stake of the coin. Suppose $p = P(\textsc{H} < th) = \frac{th}{R}$, which is the probability of winning the leader election in each iteration of the VDF. If we sybil the stake into , let's say, three coins with equal stakes, then, the probability of winning leader election for each individual stake in each iteration of VDF is $P(\textsc{H} < \frac{th}{3}) = \frac{th}{3R} = \frac{p}{3}$. This is due to the fact that difficulty parameter $th$ is proportional to the stake. Hence, the probability of winning leader election in each iteration of VDF by at least one coin is $1-\left(1-\frac{p}{3}\right)^{3}$. However, as the VDFs are iterating very fast, we are in the regime $0<p<<1$. Thus, by Binomial series expansion, $1-\left(1-\frac{p}{3}\right)^{3} = 1-(1-3\frac{p}{3} + O(p^2)) = p + O(p^2)$. Hence, this validates our design choice that difficulty parameter is proportional to the stake of the coin. This design choice is not unique to our design and is common in all longest chain based proof-of-stake protocols.

\subsection{Costless simulation attack}
\label{sec:costless_simulation_attack}
Both sleepy model of consensus \cite{sleepy} and Ouroboros Genesis \cite{badertscher2018ouroboros} have a weaker definition of dynamic availability: {\em all adversary nodes are always online starting from genesis and no new adversary nodes can join}, which makes them vulnerable to costless simulation attack as described next. In case of sleepy model of consensus \cite{sleepy}, as shown in Fig~\ref{fig:costless_simulation_sleepy}, suppose that in the $1^{st}$ year of the existence of the PoS system, only $5\%$ of the total stake, all honest, is online and actively participating in evolution of the blockchain. Consider that at the beginning of the $2^{nd}$ year, all $100\%$ of the stake is online with $20\%$ of the stake being controlled by the adversary. 
\begin{figure}
    \centering
    \includegraphics[width=0.75\textwidth]{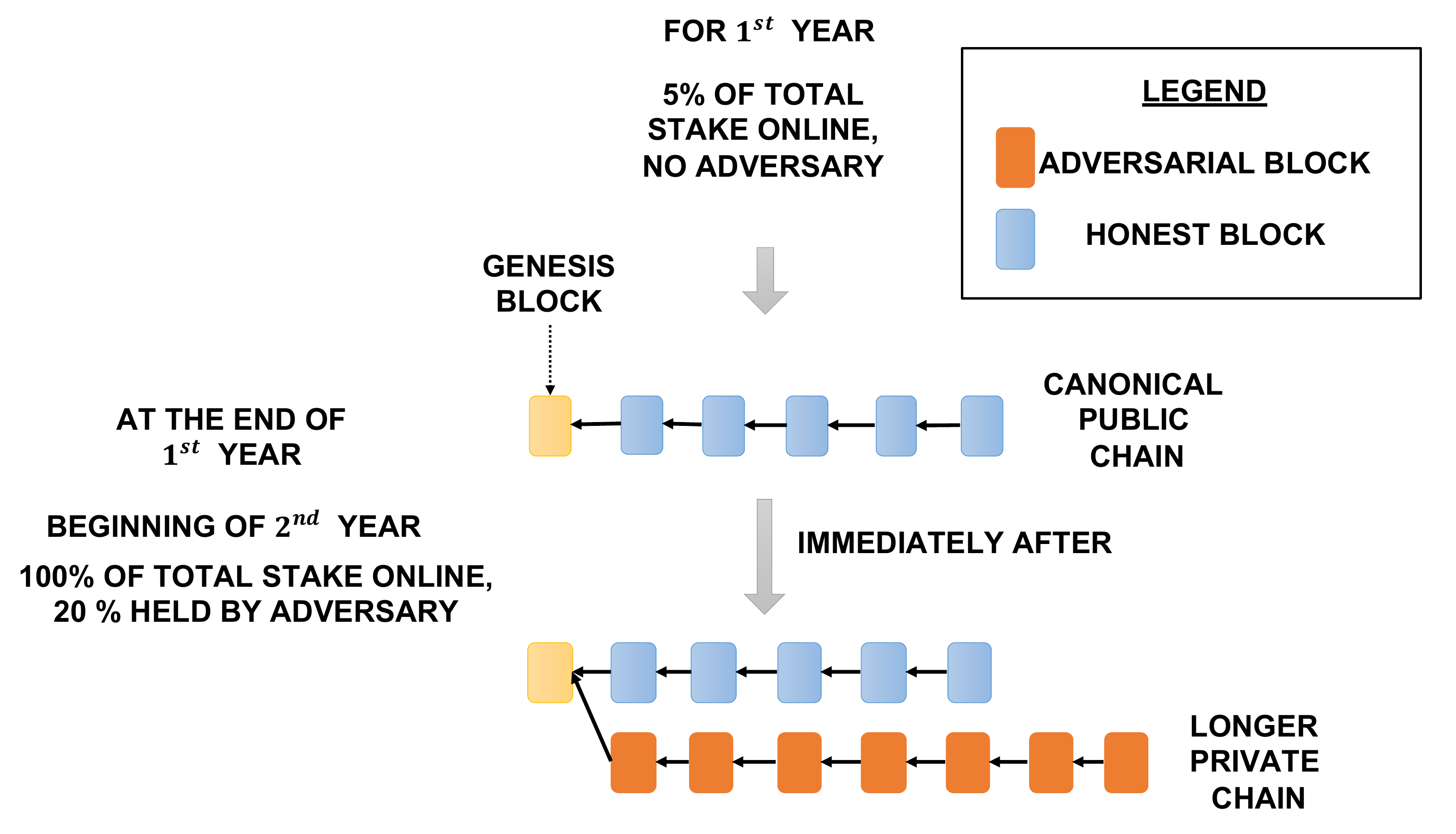}
    \caption{Costless simulation attack for sleep model of consensus \cite{sleepy}.}
    \label{fig:costless_simulation_sleepy}
\end{figure}
The adversary can costlessly simulate (with requirement of little computational time) the eligibility condition in sleepy protocol across large range of values of time $t$, thus, constructing a longer private chain than the canonical public chain. In sleepy, under the  fork-choice rule of choosing the longest chain, the private chain will be selected as the canonical chain once it is revealed by the adversary. 
\begin{figure}
    \centering
    \includegraphics[width=0.75\textwidth]{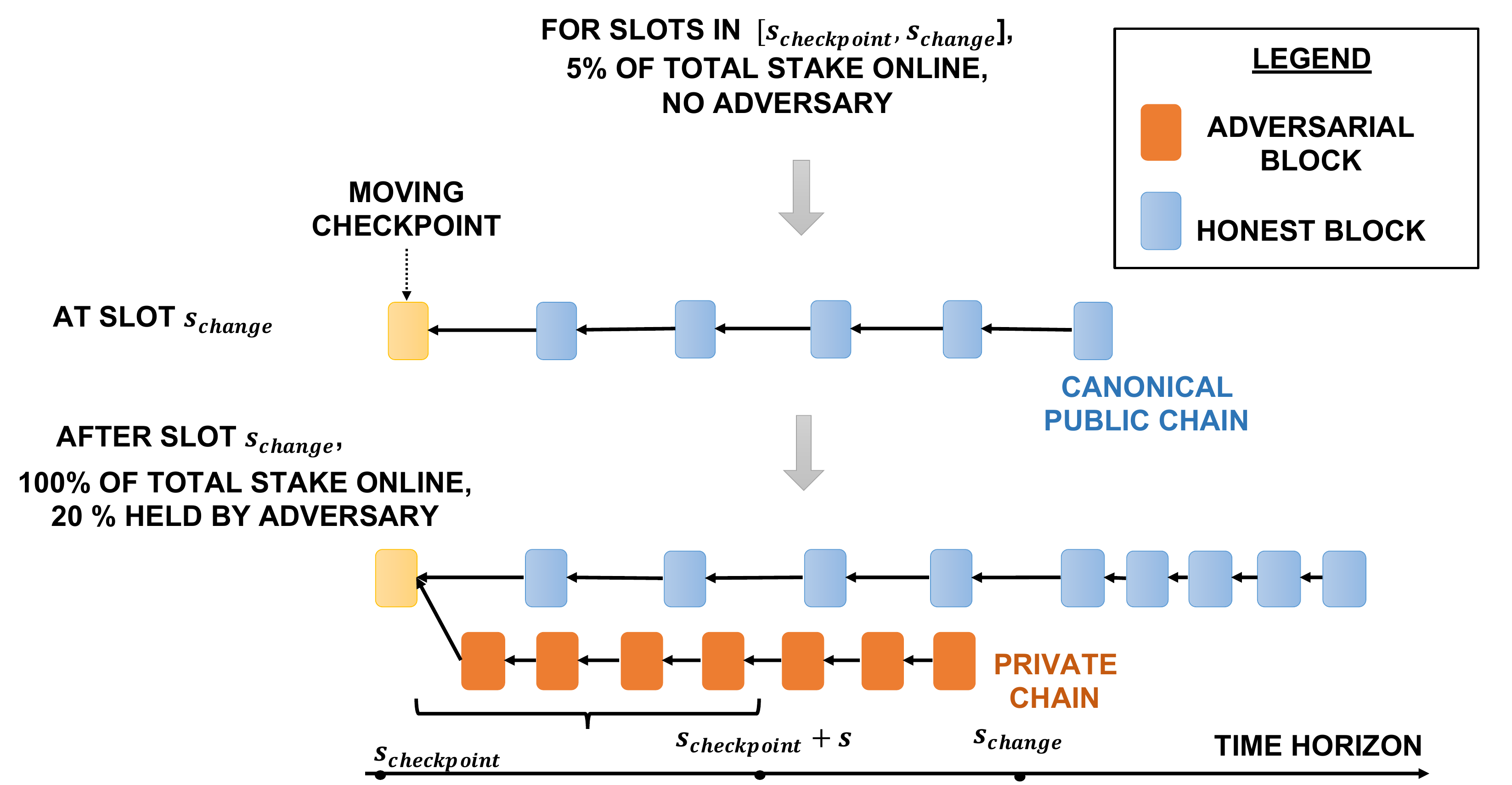}
    \caption{Costless simulation attack for Ouroboros Genesis \cite{badertscher2018ouroboros}.}
    \label{fig:costless_simulation_ouroboros}
\end{figure} 
In case of Ouroboros Genesis \cite{badertscher2018ouroboros}, as shown in Fig~\ref{fig:costless_simulation_ouroboros}, the adversary can utilize the $20\%$ stake under its control after the slot $s_{change}$ to costlessly construct a private chain involving leader elections for the slots starting from the checkpoint slot $s_{checkpoint}$. Observe that in the slots $[s_{checkpoint},s_{change}]$ of the operation of the PoS system, the canonical public chain evolved due to the participation of only $5\%$ stake. Clearly, with high probability, for any $s$ such that $s_{checkpoint} + s < s_{change}$, the private chain has more blocks in the range $[s_{checkpoint},s_{checkpoint} + s]$  as compared to canonical public chain. Under the fork-choice rule of Ouroboros Genesis as described in Fig. 10 of \cite{badertscher2018ouroboros}, the private adversarial chain will be selected as the canonical chain when it is revealed.

If the fork-choice rule is to choose the longest chain, the design recommendation for defending against costless simulation attack is to make it {\em expensive} for the adversary to propose blocks for the past slots and create longer chain. For instance, in \protocol, the adversary would have to initiate its $\rvdf$ from the first block of the epoch where $\rs$ is updated. Due to the sequential nature of the computation of $\rvdf$, with high probability, the adversary won't be able to create a private chain longer than the canonical chain.

\subsection{Bribery attack due to predictability}
\label{sec:bribery_attack}
A key property of PoW protocols is their ability to be unpredictable: no node (including itself) can know when a given node will be allowed to propose a block ahead of the proposal slot. In the existing PoS protocols, there are two notions of predictability depending on how the leader election winner is decided - {\em globally predictable} and {\em locally predictable}. 
\begin{figure}
    \centering
    \includegraphics[width=0.75\textwidth]{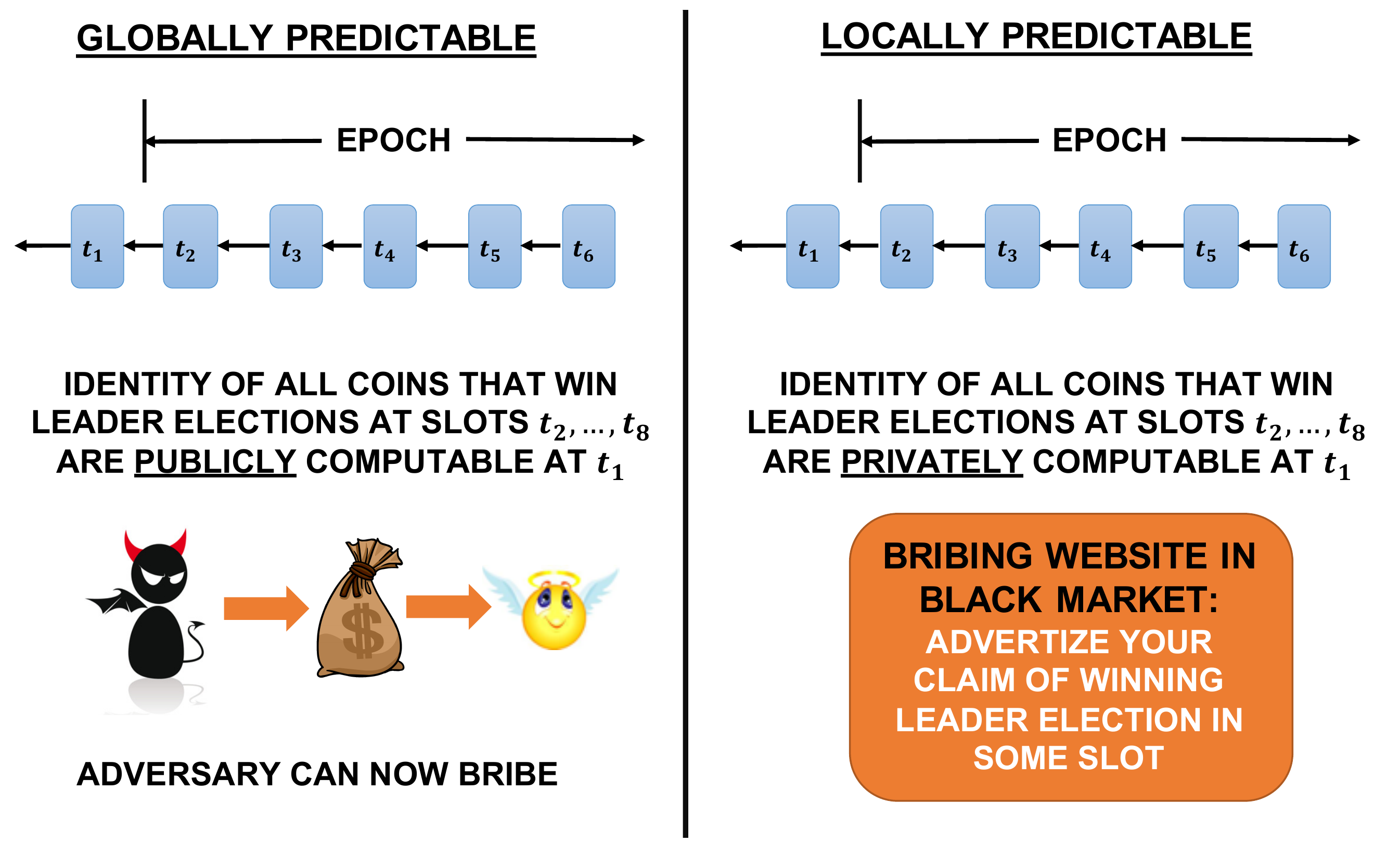}
    \caption{Variations of bribery attack stemming from predictability.}
    \label{fig:bribery_attack}
\end{figure} 
Referring to Fig~\ref{fig:bribery_attack}, using hash function for deciding winner of leader election, as in \cite{sleepy}, renders the identity of winners of leader elections in future slots publicly computable. An adversary can now bribe a coin that is going to propose a block in a future slot to include or exclude a specific transaction of adversary's choice or influence the position on where to append the block. On the other hand, using verifiable random function (VRF) for deciding winner of leader election, as in Ouroboros Praos \cite{david2018ouroboros}, Ouroboros Genesis \cite{badertscher2018ouroboros} and Snow White \cite{bentov2016snow}, mutes the aforementioned public computability. However, a node owning a coin can still locally compute the future slots in which that coin can win the leader election. Now, the node can advertise its future electability in the black market.  

The central idea on how to avoid such predictability is to ensure that a node owning a coin shouldn't learn about winning a leader election with that coin in slot $s$ before the slot $s$. In \protocol, owing to randomness of $\randiter$ in $\rvdf$, the node learns about winning a leader election for that coin in slot $s$ only after completion of sequential execution of the $\rvdf$ at slot $s$.

\subsection{Private attack by enumerating blocks within an epoch}
\label{sec:block_enumeration_attack}
In $\protocol$, at the beginning of each epoch, the $\rs$ is updated. 
\begin{figure}
    \centering
    \includegraphics[width=0.75\textwidth]{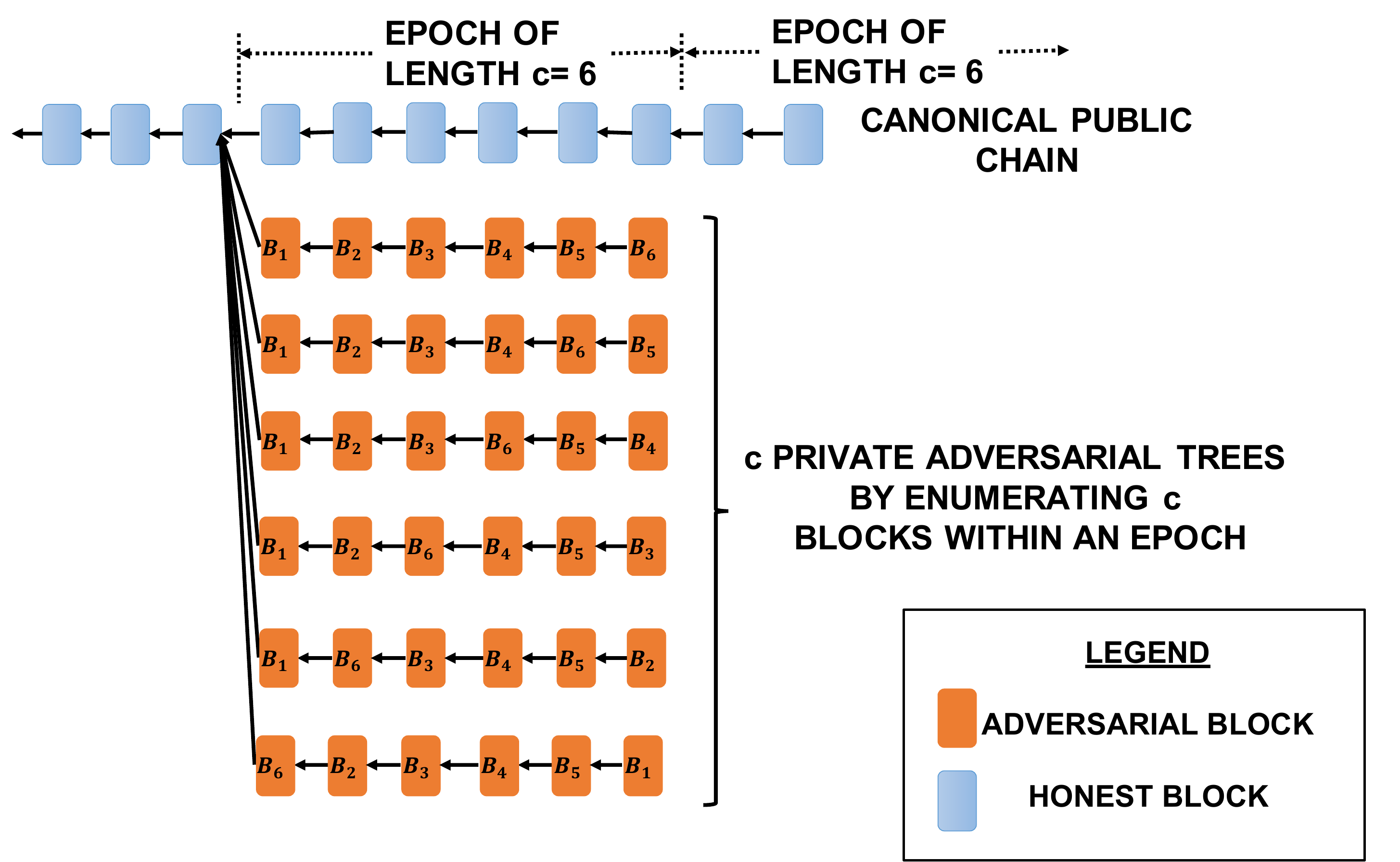}
    \caption{Enumerating blocks when time-ordering is not required.}
    \label{fig:block_enumeration_attack}
\end{figure}
However, if the appropriate guardrail in the form of time-ordering (line~\ref{algo:time-ordering} in  Algorithm~\ref{alg:PoSAT}) is not put into place, then, this $\rs$ update can provide statistical advantage to the adversary in creating longer private chain. To be specific, suppose that $\protocol$ doesn't require the $\Time$ in the blocks of a chain to be ordered in the ascending order. Then, as shown in Fig~\ref{fig:block_enumeration_attack}, the adversary can enumerate over the $c$ blocks in the private adversarial tree to have $c$ different $\rs$ updates for the next epoch. This gives $c$ distinct opportunities to the adversary to evolve the private adversarial tree which gives the aforementioned statistical advantage of order $c$ in terms of inter-arrival time of the adversarial blocks. With the guardrail of time-ordering in place, as in $\protocol$, the aforementioned enumeration is not possible as the $\Time$ contained in the blocks of a chain are required to be in ascending order.

\subsection{Long-range attack by leveraging randomness update}
\label{sec:long_range_attack}
Updating $\rs$ for a new epoch based on solely the last block of the previous epoch, as done in $\protocol$, gives rise to an unique situation in which an adversary can mount a long-range attack to create a longer private adversarial chain. Referring to Fig~\ref{fig:long_range_attack}, an adversary, with sufficiently large probability, can win at least one leader election in each epoch and publicly reveal the block associated with that leader election after appropriate delay so that the block ends up as the last block of that epoch. 
\begin{figure}
    \centering
    \includegraphics[width=0.8\textwidth]{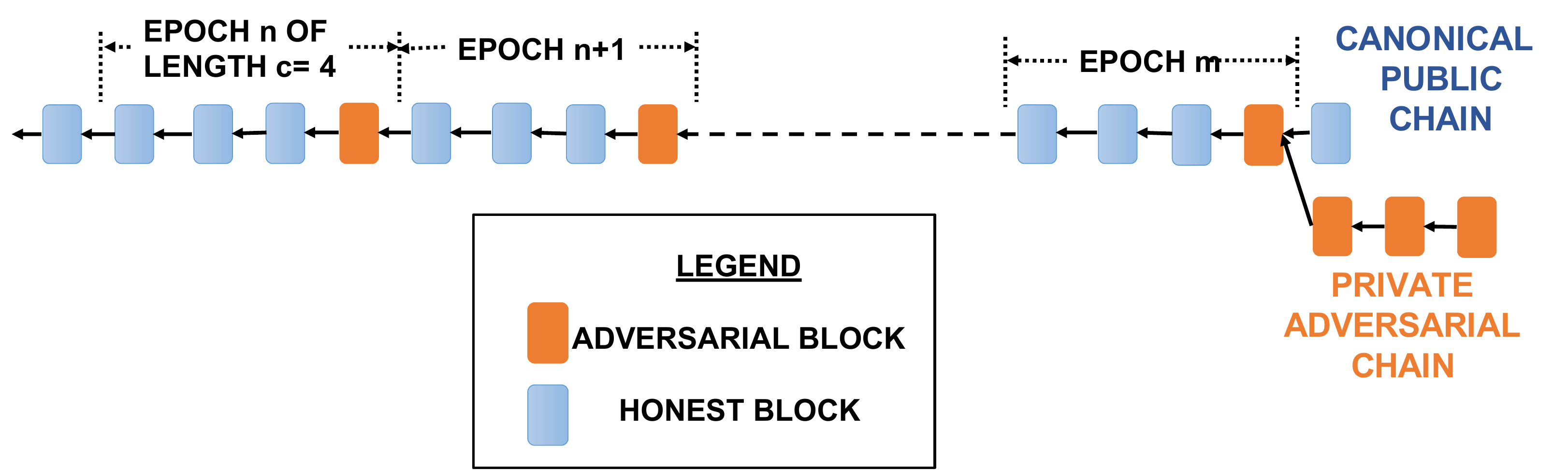}
    \caption{Illustration of the long-range attack. Consider that $m > n$.}
    \label{fig:long_range_attack}
\end{figure}
Observe that this strategy will give power to the adversary to dictate the $\rs$ for each epoch. Moreover, the adversary, on proposing on winning a leader election for an epoch, can just move on to contesting a leader election for the next epoch. Thus, the adversary can privately behave as if $c=1$ whereas the actual $c$ might be greater than $1$. With such a strategy, the adversary can win at least one leader election for many future epochs and publicly reveal the blocks associated with those leader elections in a time-appropriate manner. The adversary can continue this strategy until an appropriate epoch when it is able to win multiple leader elections and wants to do double spending. There are two design recommendations on how to protect against this long-range attack:
\begin{itemize}
    \item requiring time-ordering of the blocks in a chain, as done in $\protocol$, would ensure that even after the adversary behaves as if $c=1$ and wins leader elections for future epochs, the blocks associated with those leader elections would fail the time-ordering (line~\ref{algo:time-ordering} in  Algorithm~\ref{alg:PoSAT}). This completely removes the aforementioned long-range attack.
    \item requiring that the $\rs$ for a new epoch is dependent on all the blocks of the previous epoch. This strategy diminishes the amount of influence that an adversary can have on the $\rs$ update.
\end{itemize}

\section{Supplementary for Section~\ref{sec:primitives}}
\label{sec:vdf}
We give a brief description of VDFs, starting with its definition. 
\begin{definition}[from \cite{boneh2018verifiable}]
A VDF $V = (Setup, Eval, Verify)$ is a triple of algorithms as follows:
\begin{itemize}
    \item $Setup(\lambda,\tau) \rightarrow \mathbf{pp} = (ek, vk)$ is a randomized algorithm that takes a security parameter $\lambda$ and
    a desired puzzle difficulty $\tau$ and produces public parameters $\mathbf{pp}$ that consists of an evaluation
    key $ek$ and a verification key $vk$. We require Setup to be polynomial-time in $\lambda$. By convention, the public parameters specify an input space $\mathcal{X}$ and an output space $\mathcal{Y}$. We assume
    that $\mathcal{X}$ is efficiently sampleable. Setup might need secret randomness, leading to a scheme
    requiring a trusted setup. For meaningful security, the puzzle difficulty $\tau$ is restricted to be
    sub-exponentially sized in $\lambda$.
    \item $Eval(ek, input, \tau) \rightarrow (O, proof)$ takes an  $input \in \mathcal{X}$ and produces an output $O \in \mathcal{Y}$ and a (possibly empty) $proof$. Eval may use random bits to generate the  $proof$ but not to compute $O$. For all $\mathbf{pp}$ generated by $Setup(\lambda,\tau)$ and all $input \in  \mathcal{X}$, algorithm $Eval(ek, input, \tau)$ must run in parallel time $\tau$ with $poly(log(\tau), \lambda)$ processors.
    \item $Verify(vk, input, O, proof) \rightarrow {Yes, No}$ is a deterministic algorithm takes an input, output and proof and outputs $Yes$ or $No$. Algorithm Verify must run in total time polynomial in $log{\tau}$ and $\lambda$. Notice that $Verify$ is much faster than $Eval$.
\end{itemize}
\end{definition}

The definition for correctness and soundness for $\rvdf$ is defined as follows: 
\begin{definition}[Correctness]
A $\rvdf$ $V$ is correct if for all $\lambda, \tau$, parameters $(ek, vk) \stackrel{\$}{\leftarrow} \textsc{Setup}(\lambda)$,
and all $input \in X$ , if $(O,proof) \stackrel{\$}{\leftarrow} \textsc{Eval}(ek, input,\tau)$ then $\textsc{Verify}(vk, input, O, proof) = Yes$.
\end{definition}

\begin{definition}[Soundness]
A $\rvdf$ is sound if for all algorithms $\mathcal{A}$ that run in time $O (poly(t, \lambda))$
\begin{align*}
    Pr\left[\substack{\textsc{Verify}(vk,input,O,proof) = Yes \\
    O \neq \textsc{Eval}(ek,input,\tau)} \bigg\vert \substack{pp = (ek,vk) \stackrel{\$}{\leftarrow} \textsc{Setup}(\lambda) \\ (input,O,proof) \stackrel{\$}{\leftarrow} \mathcal{A}(\lambda,pp,\tau)}\right] \leq negl(\lambda)
\end{align*}
\end{definition}
\section{Proof of Lemma~\ref{lem:step1-lemma}}
First, we prove the following lemma.
\label{sec:step1-proof}
\begin{lemma}
\label{lem:step1-lemma-1}
Define
\begin{align}
    \label{eq:step1-error-event}
    E_1 = \{\text{There is no epoch-beginning within the time interval $[t-\sigma(c),t]$}\}
\end{align}
and, let $\kappa_0$ be the solution for the equation $\ln{\left(\frac{\lambda_{\max}}{\lambda_{\min}}(1+\kappa)\right)} = \kappa$. Then, for $\sigma(c) = (c-1)\left(\Delta + \frac{1+\kappa}{\lambda_{\min}}\right)$ and $\kappa >> \kappa_0$, we have 
$$P(E_1) \leq e^{-O(\kappa)}.$$
\end{lemma}
\begin{proof}
Define $X_d, d > 0,$ as the time it takes for $D_h$ in the original dynamic available system to reach depth $d$ after reaching depth $d-1$. Then, for some $d_0 > 0$, we have 
$
E_1  = \left\{\sum_{d=d_0}^{d_0 + c-2} X_d  > \sigma(c)\right\}.
$
Observe that, due to our blocktree partitioning, $X_d = \Delta + Y_d$, where $Y_d$ is a non-homogeneous exponential random variable. Therefore, by Chernoff bound, for any $v > 0$
\begin{align*}
    &P\left(\sum_{d=d_0}^{d_0 + c-2} X_d  > \sigma(c)\right) \leq \mathbb{E}\left(e^{v\sum_{d=d_0}^{d_0 + c-2} X_d - v\sigma(c)}\right) \\
    &\hspace{2cm} = \mathbb{E}\left(e^{v\sum_{d=d_0}^{d_0 + c-2} Y_d }\right) e^{v(c-1)\Delta - v\sigma(c)}\\
    &\hspace{2cm} \stackrel{(a)}{=} e^{v(c-1)\Delta - v\sigma(c)}\mathbb{E}_{Y_{d_0} \mid X_{d_0-1}}e^{vY_{d_0}}\cdots \mathbb{E}_{Y_{d_0+c-2} \mid X_{d_0-1}\cdots X_{d_0+c-3}}e^{vY_{d_0+c-2}}\\
    &\hspace{2cm} \stackrel{(b)}{\leq}  e^{v(c-1)\Delta - v\sigma(c)}\left(\frac{\lambda_{\max}}{\lambda_{\min} - v}\right)^{c-1} \\                                        
    &\hspace{2cm} \stackrel{(c)}{=} e^{-v\frac{(c-1)(1+\kappa)}{\lambda_{\min}}}\left(\frac{\lambda_{\max}}{\lambda_{\min} - v}\right)^{c-1} \\
\end{align*}
where $(a)$ is due to law of total expectation, $(b)$ is due to the fact that, if $f_{Y_{d_0+i} \mid X_{d_0-1}\cdots X_{d_0+i-1}}(y)$ is the pdf of $Y_{d_0+i}$ given $X_{d_0-1}\cdots X_{d_0+i-1}$, then
$\lambda_{\min} \leq \lambda_h(t) \leq \lambda_{\max}$ implies $f_{Y_{d_0+i} \mid X_{d_0-1}\cdots X_{d_0+i-1}}(y) \leq \lambda_{\max} e^{-\lambda_{\min}y}$, $(c)$ is by putting $\sigma(c) = (c-1)\left(\Delta + \frac{1+\kappa}{\lambda_{\min}}\right)$. Optimizing over $v$ implies that for $v = \lambda_{\min}\left(\frac{\kappa}{1+\kappa}\right)$, we have
\begin{align*}
    &P\left(\sum_{d=d_0}^{d_0 + c-2} X_d  > \sigma(c)\right) \leq e^{(c-1)\left[-\kappa + \ln{\left(\frac{\lambda_{\max}}{\lambda_{\min}}(1+\kappa)\right)}\right]}
\end{align*}
Note that for $\kappa >> \kappa_0$, we have $P\left(\sum_{d=d_0}^{d_0 + c-2} X_d  > \sigma(c)\right) \leq e^{-O(\kappa)}$.
\end{proof}
Recall that, under the design of simulated system $dyn2$, if an honest coin has been online in the original dynamic available system for at least time $\sigma(c)$, then the coin can also contribute to the growth of the canonical chain in $dyn2$. From Lemma~\ref{lem:step1-lemma-1}, we know that for $\sigma(c) = (c-1)\left(\Delta + \frac{1+\kappa}{\lambda_{\min}}\right)$ and $\kappa >> \kappa_0$, this honest coin has encountered at least one epoch-beginning in the original dynamic available system with probability $1-e^{-O(\kappa)}$. That implies, with high probability  $1-e^{-O(\kappa)}$, at time $t$, if an honest coin is contributing to the growth of the canonical chain in $dyn2$, then it is also contributing to the growth of the canonical chain in the original dynamical system. However, observe that at the same time $t$ in the original dynamic available system, there might be other honest coins which became online after $t-\sigma(c)$ and have encountered at least one epoch-beginning. At time $t$, these coins would contribute to the growth of the canonical chain in the original dynamic available system but won't be contributing to the growth of the canonical chain in $dyn2$. Thus, $\lambda^c_h(t) \leq \lambda_h(t)$ with probability $1-e^{-O(\kappa)}$. Consequently, $P(B_{s,s+t}) \leq P(B^{dyn2}_{s,s+t}) +  e^{-O(\kappa)}$.
\section{Proof of Lemma~\ref{lem:step2-lemma}}
\label{sec:step2-proof}
Observe that from (\ref{eq:time-relations}), we have
\begin{align}
    \label{eq:const-hon-mining} \int_{t_1}^{t_2}\lambda^c_h(t) dt  = \lambda_h\left[\alpha(t_2) - \alpha(t_1)\right]
\end{align}
Thus, $\alpha(t)$ is an increasing function in $t$. Then, we have the following lemma.
\begin{lemma}
\label{lemma:Event-Ordering}
The ordering of events in the dynamic available system $dyn2$ is same as in the static system $ss0$.
\end{lemma}
\begin{proof}
Suppose there are two events $E_1$ and $E_2$ that happen in $dyn2$ such that $t_{E_1} < t_{E_2}$, that is, $E_1$ happen before $E_2$ in $dyn2$. By contradiction, assume that $E_2$ happen before $E_1$ in the frame of reference of the static system $ss0$. By equation~\ref{eq:time-relations}, that implies, $\alpha(t_{E_2}) < \alpha(t_{E_1})$. However, this contradicts the fact that $\alpha(t)$ is an increasing function in $t$.
\end{proof}
Suppose that $B_{s,s+t}$ happens in $dyn2$. This implies that, in $dyn2$, for every honest block $b_j$ proposed at $\tau^h_j \in [s,s+t]$, there exists some minimum time $t_0 > \tau^h_j + \Delta$ and some honest block $b_i$ proposed at $\tau_i^h$ such that 
\begin{align*}
    D_i(t_0) \ge D_h(t_0-\Delta) - D_h(\tau_i^h + \Delta).
\end{align*}
Due to Lemma~\ref{lemma:Event-Ordering}, events in the evolution of the blockchain in $dyn2$ during the interval $[\tau^h_i,t_0]$ happens in the same order in the static system $ss0$ during the time interval $[\alpha(\tau^h_i),\alpha(t_0)]$. That implies the depth of the fictitious honest tree at time $t$ in the local clock of $dyn2$ is same as the depth of the same fictitious honest tree at time $\alpha(t)$ in the local clock of $ss0$. This equivalence also carries over for the adversarial trees. Then, analysing the race between the fictitious honest tree $\T_h(t)$ and the adversarial tree $\T_i(t)$ with respect to the local clock of $ss0$, we can write 
\begin{align*}
    D_i(\alpha(t_0)) \ge D_h(\alpha(t_0-\Delta)) - D_h(\alpha(\tau_i^h + \Delta))
\end{align*}
That implies $b_j$ is not a Nakamoto block in the static system $ss0$ too. Since, $b_j$ is any arbitrary honest block with $\tau^h_j \in [s,s+t]$, therefore this is true for all honest blocks $j'$ with $\tau^h_{j'} \in \{s,s+t\}$. Hence, $B^{dyn2}_{s,s+t} = B^{ss0}_{\alpha(s),\alpha(s+t)}$ which implies $P(B^{dyn2}_{s,s+t}) = P(B^{ss0}_{\alpha(s),\alpha(s+t)})$. This concludes our lemma.

\section{Proof of Lemma~\ref{lem:step3-lemma}}
\label{sec:step3-proof}
For simulating the static system $ss1$, keep the sample path of the progress of the fictitious honest tree in both static systems $ss0$ and $ss1$ same. For some $t>0$ in the local clock of $dyn2$, let $\mathcal{T}_i(t)$ represent the adversarial tree in $ss0$ with $b_i$ as its root. Suppose $B^{ss0}_{\alpha(s),\alpha(s+t)}$ happens in $ss0$ for some $s,t > 0$ defined in the local clock of $dyn2$. That implies, for every honest block $b_j$ proposed at $\alpha(\tau^h_j) \in [\alpha(s),\alpha(s+t)]$, there exists some minimum time $\alpha(t_0) > \alpha(\tau_j^h + \Delta)$ and some honest block $b_i$ proposed at $\alpha(\tau^h_i)$ such that 
$$
D_i(\alpha(t_0)) \geq D_h(\alpha(t_0 - \Delta)) - D_h(\alpha(\tau^h_i + \Delta)).
$$
Now, for any arbitrary $b_j$, there are two cases:
\begin{enumerate}
    \item If the tip of the fictitious honest tree at time $\alpha(t_0-\Delta)$ is in the same epoch as the honest block $b_i$, then, the adversary can duplicate the first block in the adversarial tree $\mathcal{T}_i(t)$ of the static system $ss0$ and attach it to the block $b_i$ of the simulated system $ss1$. However, in $ss1$, the adversary immediately gets a gift of $c-1$ blocks. 
    \item If the fictitious honest tree at time $\alpha(t_0-\Delta)$ is in a different epoch as the honest block $b_i$, then, the adversary can duplicate the $\mathcal{T}_i(t)$ and prune it to contain all the blocks starting from the epoch that comes immediately after the epoch containing $b_i$ in $\mathcal{T}_i(t)$. Then, in $ss1$, the adversary duplicates the first block in the adversarial tree $\mathcal{T}_i(t)$ of the static system $ss0$ and attaches it to the block $b_i$ of the simulated system $ss1$ that immediately gifts a chain of $c-1$ blocks. The adversary  then appends over that chain the pruned $\mathcal{T}_i(t)$.
\end{enumerate}
Both cases clearly imply that at time $\alpha(t_0)$, there is an adversarial tree on $b_i$ in $ss1$ whose depth is greater than $\mathcal{T}_i(t_0)$ in $ss0$. Thus, $b_j$ is not a Nakamoto block in $ss1$. Hence, 
$P(B^{ss0}_{\alpha(s),\alpha(s+t)}) \leq P(B^{ss1}_{\alpha(s),\alpha(s+t)})$.

\section{Growth rate of Adversarial Tree $\hat{\T}_i(t)$}
\label{sec:growth_rate}

We first give a description of the (dual of the) adversarial tree consisting of super-blocks in terms of a Branching Random Walk (BRW). 

Observe that due to the assumption on adversary in $ss2$, each adversarial tree $\hat{\mathcal{T}}_i(t')$ (with $i^{th}$ honest block as its root), when analysed in the local clock of $ss2$, grows statistically in the same way, without any dependence on the level of the root. Without loss of generality, let us focus on the adversary tree $\hat{\mathcal{T}}_0(t')$, rooted at genesis. The genesis block is always at depth $0$ and hence $\hat{\mathcal{T}}_0(0)$ has depth zero. 

We can transform the tree $\hat{\T}_0(t')$ into a new random tree $\hat{\T}^s_0(t')$. Every $c$ generations in $\hat{\T}_0(t')$ can be viewed as a single generation in $\hat{\T}^s_0(t')$. Thus, every block in $\hat{\T}^s_0(t')$, termed as \textit{superblocks}, is representative of $c$ blocks in $\hat{\T}_0(t')$. Consider $B_0$ to be the root of $\hat{\T}^s_0(t')$. The children blocks of $B_0$ in $\hat{\T}^s_0(t')$ are the descendent blocks at level $c$ in $\hat{\T}_0(t')$. We can order these
children blocks of $B_0$ in terms of their arrival times. Then, as the blocks in first $c-1$ levels of $\hat{\T}_0(t')$ are gift, the adversary didn't have to compute $\rvdfeval$ for these blocks. Consider block $B_1$ to be the first block for which $\rvdfeval$ was computed by the adversary.  Therefore, the arrival time $Q_1$ of block $B_1$ is given by $X_1$ where $X_1$ is an exponential random variable in the static system $ss2$. On the other hand, the arrival time of the first child of $B_1$, call it $B_{1,1}$, is given by $Q_{1,1} = Q_1 + X_{1,1} + \cdots + X_{1,c}$, where $X_{1,i}$ is the inter-arrival time between the $(i-1)^{th}$ and $i^{th}$ descendent block of $B_1$. Note that, in the static system $ss2$, all the $X_{1,i}$'s are exponential with parameter $\lambda_a$, and they all are independent. Let the depth of the tree $\hat{\T}^s_0(t')$ be $D^s_0(t')$. 

Each vertices at generation $k\geq 2$ in $\hat{\T}^s_0(t')$ can be labelled as a $k$ tuple of positive integers $(i_1,\ldots,i_k)$ with $i_j \geq c$ for $2\leq j \leq k$: the vertex $v = (i_1,\ldots,i_k)\in \I_k$ 
is the $(i_k-c+1)$-th child of vertex $(i_1,\ldots,i_{k-1})$
at level $k-1$. At $k=1$ generation, we have $i_1 \geq 1$ as the adversary is gifted $c-1$ blocks on proposing the first block for which it computes only one $\rvdfeval$.
Let $\I_k =\{(i_1,\ldots,i_k): i_j \geq 1 \text{ for } i_j = 1 \text{ and } i_j \geq c \text{ for } 2\leq j \leq k\}$, and  set $\I=\cup_{k>0}  \I_k$.  For such $v$ we also let $v^j=(i_1,\ldots,i_j)$,  $j=1,\ldots,k$, denote the  ancestor of $v$ at level $j$, with $v^k=v$. For notation convenience, we set $v^0=0$ as the root of the tree.

Next, let $\{\E_v\}_{v\in \I}$ be an i.i.d. family of exponential random variables of parameter $\lambda_a$. For $v=(i_1,\ldots,i_k)\in \I_k$,
let $\W_v=\sum_{j\leq i_k} \E_{(i_1,\ldots,i_{k-1},j)}$ and
let $Q_v=\sum_{j\leq k} \W_{v^j}$. This creates a labelled tree, with the following
interpretation: for $v=(i_1,\ldots,i_j)$,  the $W_{v^j}$ are the waiting time
for $v^j$ to appear, measured from the appearance of $v^{j-1}$, and 
$Q_v$ is the appearance time of $v$. Observe that starting from any $v \in \I_1$, we obtain a standard BRW. For any $v = (i_1, \cdots, i_k) \in \I_k$, we can write $Q_v = Q^1_v + Q^2_v$ where $Q^1_v$ is the appearance time for the ancestor of $v$ at level $1$ while $Q^2_v = Q_v - Q^1_v$. 

Let $Q^*_k=\min_{v\in \I_k} Q_v$. Note that $Q^*_k$ is the time of appearance of a block at level $k$ and therefore we have
\begin{equation}
  \label{eq:c_grow}
  \{D_0(t') \geq ck\} = \{D_0^s(t') \geq k\} = \{Q^*_k\leq  t'\}.
\end{equation}
Fixing $i_1 \in \mathcal{I}_1$, let  $Q^{2*}_{k,i_1}=\min_{v\in \I_k \text{ s.t. }v^1 = i_1} Q^2_v$. Observe that $Q^{2*}_{k,i_1}$ is the minimum of a standard BRW with its root at the vertex $i_1$. Introduce, for $\theta_c<0$,
the moment generating 
function
\begin{eqnarray*}
\Lambda_c(\theta_c)&=&\log \sum_{\substack{v\in \I_2\\v^1 = i_1}} E(e^{\theta_c Q^2_v})=
\log \sum_{j=c}^\infty E (e^{\sum_{i=1}^j \theta_c \E_i})\\
&=&\log \sum_{j=c}^\infty (E(e^{\theta_c \E_1}))^j  
=\log \frac{E^c(e^{\theta_c \E_1})}{1-E(e^{\theta_c\E_1})}.
\end{eqnarray*}
Due to the exponential law of $\E_1$,
$E(e^{\theta_c \E_1})= \frac{\lambda_a}{\lambda_a-\theta_c}$ and therefore
$\Lambda_c(\theta_c)=\log(-\lambda_a^c/\theta_c(\lambda_a-\theta_c)^{c-1})$. 

An important role is played by $\theta_c^*$, which is the negative solution to the equation $\Lambda_c(\theta_c) = \theta_c {\dot \Lambda_c(\theta_c)}$ and let $\eta_c$ satisfy that
$$\sup_{\theta_c<0} \left(\frac{\Lambda_c(\theta_c)}{\theta_c}\right)= \frac{\Lambda_c(\theta_c^*)}{\theta_c^*}=\frac{1}{\lambda_a \eta_c}.$$
Indeed,  we have the following.
\begin{proposition}
  \label{prop:prop-c}
$$\lim_{k\to\infty} \frac{Q^*_k}{k} = \lim_{k\to\infty} \frac{Q^{2*}_{k,i_1}}{k} = \sup_{\theta_c<0} \left(\frac{\Lambda_c(\theta_c)}{\theta_c}\right)=\frac{1}{\lambda_a \eta_c}, \quad a.s.$$
\end{proposition}
In fact, much more is known. 
\begin{proposition}
  \label{prop:prop-c_tight}
  There exist explicit constants $c_1>c_2>0$ 
so that the sequence
$ Q^*_k-k/\lambda_a \eta_c-c_1\log k$ is tight,
and
$$\liminf_{k\to \infty} Q^*_k-k/\lambda_a \eta_c-c_2\log k=\infty, a.s.$$
\end{proposition}

Note that Propositions \ref{prop:prop-c},\ref{prop:prop-c_tight}
and \eqref{eq:c_grow} imply in particular
that  $D_0(t')\leq c\eta_c\lambda_a t'$ for all large $t'$, a.s., and also that
\begin{equation}
  \label{eq:c_converse}
\mbox{\rm  if
$c\eta_c\lambda_a>\lambda_h$ then } D_i(t')>t'\; \mbox{\rm for all large $t'$, a.s.}.
\end{equation}

Let us define $\phi_c := c\eta_c$, then $\phi_c \lambda_a$ is the growth rate of private $c$-correlated NaS tree. With all these preparations, we can give a simple proof for
Lemma~\ref{lem:step4-lemma}.
\begin{proof}
  Consider $m=\eta_c\lambda_a t'+x$.
  Note that by \eqref{eq:c_grow},
  \begin{align}
    P(D_0^s(t')\geq m)&=P(Q^*_{ m}\leq t')
    \leq \sum_{v\in \I_{m}} P(Q_v\leq t') = \sum_{v\in \I_{m}}  P(Q^1_v + Q^2_v \leq t') \nonumber  \\
    &= \sum_{v\in \I_{m}} \int_0^{t'} p_{Q^1_v}(u) P(Q^2_v \leq t' - u) du \nonumber \\
    \label{eq:c-bound} 
    &= \sum_{i_1 \geq 1} \sum_{i_2\geq c,\ldots,i_{m}\geq c} \int_0^{t'} p_{Q^1_v}(u) P(Q^2_v \leq t' - u) du
\end{align}
  For $v=(i_1,\ldots,i_k)$, set $|v_{-1}|=i_2+\cdots+i_k$. Then,
  we have that $Q^2_v$ has the same law as 
  $\sum_{j=1}^{|v_{-1}|} \E_j$. Thus, by Chebycheff's inequality,
  for $v\in \I_{m}$,
  \begin{equation}
    \label{eq:upper_1}
    P(Q^2_v\leq  t'-u)\leq Ee^{\theta_c^* Q^2_v} e^{-\theta_c^*(t'-u)}=\left(\frac{\lambda_a}
  {\lambda_a-\theta_c^*} \right)^{|v_{-1}|} e^{-\theta_c^* (t'-u)}.
\end{equation}
  And
  \begin{eqnarray}
    \label{eq:upper_2}
    \sum_{i_2\geq c,\ldots,i_{m}\geq c}\left(\frac{\lambda_a}
  {\lambda_a-\theta_c^*} \right)^{|v_{-1}|}
  &=&\sum_{i_2\geq c,\ldots,i_{m}\geq c} 
  \left(\frac{\lambda_a}
  {\lambda_a-\theta_c^*} \right)^{\sum_{j=2}^m i_j}\nonumber \\
  &=&\left(\sum_{i\geq c}\left(\frac{\lambda_a}
  {\lambda_a-\theta_c^*} \right)^i\right)^{m-1}= e^{(m-1)\Lambda_c(\theta_c^*)}.
\end{eqnarray}
Combining \eqref{eq:upper_1}, \eqref{eq:upper_2} and \eqref{eq:c-bound}
yields
\begin{eqnarray}
P(D_0^s(t') \ge m) &\leq& e^{-\theta_c^* t'} e^{(m-1)\Lambda_c(\theta_c^*)} \sum_{i_1 \ge 1}  \int_{0}^{t'} p_{Q^1_v}(u)e^{\theta_c^* u} du \nonumber \\
&=&  e^{-\theta_c^* t'} e^{(m-1)\Lambda_c(\theta_c^*)} \sum_{i_1 \ge 1} \int_{0}^{t'} \frac{\lambda_a^{i_1}u^{i_1 - 1}e^{-\lambda_a u}}{\Gamma(i_1)} e^{\theta_c^* u} du \nonumber \\
&=& e^{-\theta_c^* t'}e^{(m-1)\Lambda_c(\theta_c^*)}g(t').
\end{eqnarray}
where $g(t') = \sum_{i_1 \ge 1} \int_{0}^{t'} \frac{\lambda_a^{i_1}u^{i_1 - 1}e^{-\lambda_a u}}{\Gamma(i_1)} e^{\theta_c^* u} du$.
\end{proof}
From proposition~\ref{prop:prop-c}, we have
\begin{equation}
    \label{eq:phi_c}
    \phi_c = \frac{c\theta_c^*}{\lambda_a} \left(\frac{1}{\log{\left(\frac{-\lambda^c_a}{\theta^*_c(\lambda_a - \theta^*_c)^{c-1}}\right)}}\right),
\end{equation}
where $\theta_c^*$ is the unique negative solution of 
\begin{align}
    \Lambda_c(\theta_c) &= \theta_c {\dot \Lambda_c(\theta_c)} 
    \label{eqn:relation}
\end{align}
Note that $g(t')$ is an increasing function on $t'$ and
\begin{align}
    \label{eq:bound-g-t}
    \lim_{t' \rightarrow \infty} g(t') = \sum_{i_1 \geq 1} \left(\frac{\lambda_a}{\lambda_a - \theta_c^{\star}}\right)^{i_1} = \frac{\lambda_a}{- \theta_c^{\star}}
\end{align}

\section{Proofs}
\subsection{Definitions and Preliminary Lemmas}
\label{sec:preliminary-lemmas}
In this section, we define some important events which will appear frequently in the analysis and provide some useful lemmas.

Let $V_j$ be the event that the $j-$th honest block $b_j$
is a loner, i.e.,
\begin{align*}
    V_j = \{\tau^h_{j-1} < \tau^h_{j} - \Delta'\} \bigcap \{\tau^h_{j+1} > \tau^h_{j} + \Delta'\}
\end{align*}
\noindent Let $\hat{F}_j = V_j \bigcap F_j$ be the event that $b_j$ is a Nakamoto block. Then, we can define the following ``potential" catch up event in $ss2$: 
\begin{equation}
\label{eqn:hat_Bik}
    \hat{B}_{ik} = \{D_i(\tau^h_{k} + \Delta')  \ge  D_h(\tau^h_{k-1}) - D_h(\tau^h_i+\Delta')\},
\end{equation}
which is the event that the adversary launches a private attack starting from honest block $b_i$ and catches up the fictitious honest chain right before honest block $b_k$ is mined. 
\begin{lemma}  
\label{lem:F_j}
For each $j$,
\begin{equation}
    P(\hat{F}_j^c) = P(F_j^c \cup V_j^c) \leq P\left(\left(\bigcup_{(i,k): 0 \leq i < j <k} \hat{B}_{ik}\right) \cup V_j^c\right).
\end{equation}
\end{lemma}
\begin{proof}
\begin{eqnarray*}
&& P(V_j \cap E_{ij}) \\
&=& P(V_j \cap \mbox{\{$D_i(t')  < D_h(t'-\Delta') - D_h(\tau^h_i+\Delta')$ for all $t' > \tau^h_j + \Delta'$\}})\\
&= & P(V_j \cap \mbox{$\{D_i(t' +\Delta')  < D_h(t') - D_h(\tau^h_i+\Delta')$ for all $t' > \tau^h_j\}$ })\\
&= & P(V_j \cap \mbox{$\{D_i({\tau^h_k}^- +\Delta')  < D_h({\tau^h_k}^-) - D_h(\tau^h_i+\Delta')$ for all $k > j\}$}) \\
&= & P(V_j \cap \mbox{$\{D_i({\tau^h_k} +\Delta')  < D_h({\tau^h_{k-1}}) - D_h(\tau^h_i+\Delta')$ for all $k > j\}$}) .
\end{eqnarray*}
Since $\hat{F}_j = F_j \cap V_j = \bigcap_{0 \leq i < j} E_{ij} \cap V_j$, by the definition of $\hat{B}_{ik}$ we have $P(\hat{F}_j) \geq P\left(\left(\bigcap_{(i,k): 0 \leq i < j <k} \hat{B}_{ik}^c\right) \cap V_j\right)$. Taking complement on both side, we can conclude the proof.
\end{proof}

Let $R_m = \tau^h_{m+1} - \tau^h_{m}$. Then, $V_j$ and $\hat{B}_{ik}$ can be re-written as:
\begin{align}
    V_j &= \{\Delta' < R_{j-1}  \} \bigcap \{R_{j} > \Delta'\}\\
    \hat{B}_{ik} &=  \Bigg\{D_i(\tau^h_i+ \sum_{m=i}^{k-1}R_m + \Delta')  \ge  D_h(\tau^h_{k-1}) - D_h(\tau^h_i + \Delta')\Bigg\} \nonumber 
\end{align}
\begin{remark}
\label{rem:R_m}
By time-warping, $R_m$ is an IID exponential random variable with rate $\lambda_h$.
\end{remark}

Define $X_d$, $d > 0$, as the time it takes in the local clock of static system $ss2$ for $D_h$ to reach depth $d$ after reaching depth $d-1$. In other words, $X_d$ is the difference between the times $t_1$ and $t_2$, where $t_1$ is the minimum time $t'$ in the local clock of $ss2$ such that $D_h(t') = d$, and, $t_2$ is the minimum time $t'$ in the local clock of $ss2$ such that $D_h(t') = d-1$.

Also, let $\delta^h_j = \tau^h_j - \tau^h_{j-1}$ and $\delta^a_j = \tau^a_j - \tau^a_{j-1}$ denote the inter-arrival time for honest and adversary arrival events in the local clock of static system $ss2$, respectively.

\begin{proposition}
\label{prop:min-rate}
Let $Y_d$, $d \geq 1$, be i.i.d random variables, exponentially distributed with rate $\lambda_h$. Then, each random variable $X_d$ can be expressed as $\Delta'+Y_d$.
\end{proposition}
See Proposition $C.1$ in \cite{dembo2020everything} for the proof.

\begin{proposition}
\label{prop:bound-1}
For any constant $a$, 
$$P(\sum_{d=a}^{n+a} X_d > n(\Delta' + \frac{1}{\lambda_h})(1+\delta)) \leq e^{-n\Omega(\delta^2(1+\Delta'\lambda_h)^2)}$$
\end{proposition}
Proposition~\ref{prop:bound-1} is proved using chernoff bound and Proposition~\ref{prop:min-rate}.

\begin{proposition}
\label{prop:bound-2}
Probability that there are less than
$$n\frac{\lambda_a(1-\delta)}{\lambda_h}$$ 
adversarial arrival events for which $\rvdfeval$ has been computed in the interval $\tau^h_0$ to $\tau^h_{n+1}$ is upper bounded by
$$e^{-n\Omega(\delta^2\frac{\lambda_a}{\lambda_h})}$$
\end{proposition}
\noindent Proposition~\ref{prop:bound-2} is proven using the Poisson tail bounds.

\begin{proposition}
\label{prop:bound-3}
For $n > \frac{c-1}{\phi_c-1}$, define $B_n$ as the event that there are at least $n$ adversarial block arrivals for each of which adversary computed $\rvdfeval$ while $D_h$ grows from depth $0$ to $n+c-1$:
$$B_n = \{\sum_{i=1}^{n+c-1} X_i \geq \sum_{i=0}^n \delta^a_i \}$$
If 
$$\phi_c\lambda_a < \frac{\lambda_h}{1+\lambda_{h} \Delta'},$$
then,
$$P(B_n) \leq  e^{-A_1 n}e^{-A_2}$$,
\begin{align*}
    A_1 &= -w \Delta' + \ln{\left(\frac{\lambda_a+w}{\lambda_a}\right)} + \ln{\left(\frac{\lambda_h-w}{\lambda_h}\right)}\\
    A_2 &=-(c-1)w\Delta' + (c-1)\ln{\left(\frac{\lambda_h-w}{\lambda_h}\right)}
\end{align*}
such that $A_1 + \frac{A_2}{n} > 0$
and,                                                                                                                                                                                                      
\begin{align*}
    w &= \frac{\lambda_h - \lambda_a}{2} + \frac{2n+c-1}{2(n+c-1)\Delta'} -\\
    & \frac{\sqrt{[(n+c-1)\Delta'(\lambda_a-\lambda_h)]^2 + (2n+c-1)^2+2(n+c-1)\Delta'[(c-1)(\lambda_a+\lambda_h)+2(n+c-1)\Delta'\lambda_a\lambda_h]}}{2(n+c-1)\Delta'}
\end{align*}
\end{proposition}

\begin{proof}
Using Chebychev inequality and proposition~\ref{prop:min-rate}, for any $t>0$, we have 
\begin{align*}
    P(B_n) &\leq E\left[\prod_{j=0}^n e^{-w\delta_i^a}\right]E\left[\prod_{j=1}^{n+c-1} e^{wX_i}\right] \\
    &\leq \left[\frac{\lambda_a}{\lambda_a+w}\right]^{n}\left[\frac{e^{w\Delta'}\lambda_h}{\lambda_h-w}\right]^{n+c-1}\\
    &= e^{-n\left[-\left(\frac{n+c-1}{n}\right)w\Delta' + \left(\frac{n+c-1}{n}\right)\ln\left(\frac{\lambda_h-w}{\lambda_h}\right) + \ln{\left(\frac{\lambda_a+w}{\lambda_a}\right)}\right]}
\end{align*}
Optimizing over $w$, we have 
\begin{align*}
    &\frac{d}{dw} \left[-\left(\frac{n+c-1}{n}\right)w\Delta' + \left(\frac{n+c-1}{n}\right)\ln\left(\frac{\lambda_h-w}{\lambda_h}\right) + \ln{\left(\frac{\lambda_a+w}{\lambda_a}\right)}\right] = 0 \\
    &(n+c-1)\Delta' w^2 + [(n+c-1)\Delta'(\lambda_a - \lambda_h) - (2n+c-1)]w \\
    &\hspace{3cm}+ [n\lambda_h - (n+c-1)\lambda_a - (n+c-1)\Delta'\lambda_a\lambda_h] = 0\\
    &w = \frac{\lambda_h - \lambda_a}{2} + \frac{2n+c-1}{2(n+c-1)\Delta'} -\\
    &\quad \frac{\sqrt{[(n+c-1)\Delta'(\lambda_a-\lambda_h)]^2 + (2n+c-1)^2+2(n+c-1)\Delta'[(c-1)(\lambda_a+\lambda_h)+2(n+c-1)\Delta'\lambda_a\lambda_h]}}{2(n+c-1)\Delta'}
\end{align*}
Note that for $n > \frac{c-1}{\phi_c-1}$, we have $\lambda_a\left(1 + \frac{c-1}{n}\right) < \phi_c \lambda_h < \frac{\lambda_h}{1+\Delta'\lambda_h}$. That implies $w > 0$.

Also, using $n > \frac{c-1}{\phi_c-1}$, we have 
\begin{align*}
    -\left(\frac{n+c-1}{n}\right)w\Delta' + \left(\frac{n+c-1}{n}\right)\ln\left(\frac{\lambda_h-w}{\lambda_h}\right) + \ln{\left(\frac{\lambda_a+w}{\lambda_a}\right)} = A_1 + \frac{A_2}{n} > 0
\end{align*}

\end{proof}

\begin{lemma}
\label{lem:PBik_pow}
For $k-i > \frac{\lambda_h(c-1)}{\lambda_a(\phi_c-1)}$, there exists a constant $\gamma>0$ such that
\begin{align}
\label{eqn:PBik}
P(\hat{B}_{ik}) \leq e^{-\gamma (k-i)}
\end{align}
\end{lemma}
\begin{proof}

Let $N(\tau^h_i, \tau^h_k + \Delta')$ be the number of adversarial arrivals  for which $\rvdfeval$ in was computed in $ss2$ in the interval $[\tau^h_i, \tau^h_k + \Delta']$. Define
$$
\hat{C}_{ik} = \text{ event that } N(\tau^h_i, \tau^h_k + \Delta')  + (c-1) \geq D_h(\tau^h_{k-1}) - D_h(\tau^h_i + \Delta')
$$
Observe that $D_i(\tau^h_i, \tau^h_k + \Delta') \leq N(\tau^h_i, \tau^h_k + \Delta') + (c-1) $, where $c-1$ is due to the fact that blocks in first $c-1$ levels are gifted to the adversary on proposing the first block in the adversarial tree. Note that $\rvdfeval$ was not computed by the adversary for these $c-1$ blocks. Then, we have $$\hat{B}_{ik} \subseteq \hat{C}_{ik}.$$
\begin{eqnarray}
P(\hat{B}_{ik}) &\leq& P\left(N(\tau^h_i, \tau^h_k + \Delta') < (1-\delta)(k-i)\frac{\lambda_a}{\lambda_h}\right) \nonumber\\
&\quad+& P\left(\hat{C}_{ik} \mid N(\tau^h_i, \tau^h_k + \Delta') \ge (1-\delta)(k-i)\frac{\lambda_a}{\lambda_h} \right) \nonumber \\
&\stackrel{(a)}{\leq}& e^{-\Omega((k-i)\delta^2\lambda_a/\lambda_h)} + P\left(\hat{C}_{ik} \mid N_a(\tau^h_i, \tau^h_k + \Delta') \ge (1-\delta)(k-i)\frac{\lambda_a}{\lambda_h} \right) \nonumber\\
&\stackrel{(b)}{\leq}& e^{-\Omega((k-i)\delta^2\lambda_a/\lambda_h)} + \nonumber \\
&\quad& \sum_{x = (1-\delta)(k-i)\frac{\lambda_a}{\lambda_h}}^{\infty} P\left(D_h(\tau^h_{k-1}) - D_h(\tau^h_i+\Delta') \leq x + c-1 \mid N_a(\tau^h_i, \tau^h_k + \Delta') = x\right) \nonumber \\
&\stackrel{(c)}{=}& e^{-\Omega((k-i)\delta^2\lambda_a/\lambda_h)} + \sum_{x = (1-\delta)(k-i)\frac{\lambda_a}{\lambda_h}}^{\infty}e^{-A_1x}e^{-A_2}\nonumber \\
&\stackrel{(d)}{=}& e^{-\Omega((k-i)\delta^2\lambda_a/\lambda_h)} + e^{-A_2}\frac{1}{1-e^{-A_3}}e^{-A_3(k-i)}\nonumber
\end{eqnarray}
where $(a)$ is due to proposition \ref{prop:bound-2} which says that there are more than $(1-\delta)(k-i)\lambda_a/\lambda_h$  adversarial arrival events in the time period $[\tau^h_i,\tau^h_{k}+\Delta']$  except with probability $e^{-\Omega((k-i)\delta^2\lambda_a/\lambda_h)}$, $(b)$ is by union bound, $(c)$ is by proposition \ref{prop:bound-3} for $k-i > \frac{\lambda_h(c-1)}{\lambda_a(\phi_c-1)}$, (d) is due to $A_3 = 
\frac{A_1(1-\delta)\lambda_a}{\lambda_h}$.

Hence,
\begin{align}
    P(\hat{B}_{ik}) < C_1 e^{-C_2(k-i)} 
\end{align}
for appropriately chosen constants $C_1,C_2,>0$ as functions of the fixed $\delta$. Finally, since $P(\hat{B}_{ik})$ decreases as $k-i$ grows and is smaller than $1$ for sufficiently large $k-i$, we obtain the desired inequality for a sufficiently small $\gamma\leq C_3$.

\end{proof}

\subsection{Proof of Lemma~\ref{lem:step5-lemma-1}}
\label{sec:nakamoto_block}
For notational convenience, we will continue to use $\tau_i^h$ and $\tau_i^a$ as the arrival time of the $i-th$ honest and adversarial blocks in the static system $ss2$, respectively. In this proof, let $r_h := \frac{\lambda_h}{1+\lambda_{h} \Delta'}$. The random processes of interest start from time $0$. To look at the system in stationarity, let us extend them to $-\infty < t' < \infty$. More specifically, define $\tau^h_{-1}, \tau^h_{-2}, \ldots$ such that together with $\tau^h_0, \tau^h_1, \ldots$, we have a double-sided infinite random process. Also, for each $i < 0$, we define an independent copy of a random adversary tree $ \hat{\tt}_i$ with the same distribution as $\hat{\tt}_0$.  And we extend the definition of $\hat{\T}_h(t')$ and $D_h(t')$ to $t' < 0$: the last honest block mined at $\tau^h_{-1} < 0$ and all honest blocks mined within $(\tau^h_{-1}-\Delta',\tau^h_{-1})$ appear in $\hat{\T}_h(t')$ at their respective mining times to form the level $-1$, and the process repeats for level less than $-1$; let $D_h(t')$ be the level of the last honest arrival before $t'$ in $\hat{\T}_h(t')$, i.e., $D_h(t') = \ell$ if $\tau^h_i \leq t' < \tau^h_{i+1}$ and the $i$-th honest block appears at level $\ell$ of $\hat{\T}_h(t)$. 

These extensions allow us to extend the definition of $E_{ij}$ to all $i,j$, $-\infty < i< j < \infty$, and define $E_j$ and $\hat{E}_j$ to be:
$$ E_j = \bigcap_{i < j} E_{ij}$$ and 
$$ \hat{E}_j = E_j \cap V_j.$$

Note that $\hat{E}_j \subset \hat{F}_j$, so to prove that $\hat{F}_j$ has a probability bounded away from $0$ for all $j$, all we need is to prove that $\hat{E}_j$ has a non-zero probability.

Recall that we have defined the events $V_j$ and $\hat{B}_{ik}$ in section~\ref{sec:preliminary-lemmas} of the appendix as:
\begin{align*}
    V_j &= \{\Delta' < R_{j-1}  \} \bigcap \{R_{j} > \Delta'\}\\
    \hat{B}_{ik} &=  \Bigg\{D_i(\tau^h_i+ \sum_{m=i}^{k-1}R_m + \Delta')  \ge  D_h(\tau^h_{k-1}) - D_h(\tau^h_i + \Delta')\Bigg\} \nonumber 
\end{align*}
where $R_m$ are i.i.d exponential random variable with mean $\frac{1}{\lambda_h}$. 

Following the idea in Lemma~\ref{lem:F_j} and using Lemma~\ref{lem:event-reduction-1} and \ref{lem:event-reduction-2}, we have 
\begin{align*}
  P(E_j \cap V_j) = P\left(\bigcap_{i < j} E_{ij} \cap V_j\right) =  P\left(\left(\bigcap_{i < j <k} \hat{B}_{ik}^c\right) \cap U_j\right).  
\end{align*}
where $E_j = \bigcap_{i < j <k} \hat{B}_{ik}^c$ and $\hat{E}_j = E_j \cap U_j$. So, we just need to prove that $\hat{E}_j$ has a non-zero probability. Observe that, due to constant adversarial and honest mining rate  and the growth rate of the adversarial tree being independent of level of its root in the static system $ss2$, $\hat{E}_j $ has a time-invariant dependence on $\{\zz_i\}$, which means that $p = P(\hat{E}_j )$ does not depend on $j$. Then we can just focus on $P(\hat{E}_0)$. This is the last step to prove.
\begin{eqnarray*}
P(\hat{E}_0) &=& P(E_0|U_0)P(U_0)  \\
 &=& P(E_0|U_0)P(R_0>\Delta')P(R_{-1}>\Delta') \\
 &=& e^{-2\lambda_h \Delta'} P(E_0|U_0).
\end{eqnarray*}
where we used Remark~\ref{rem:R_m} in the last step. It remains to show that $P(E_0|U_0)>0$. We have
\begin{eqnarray*}
    E_0 &=& \mbox{event that } D_i(\sum_{m = i}^{k-1} R_m + \Delta' + \tau^h_i)  < D_h(\tau^h_{k-1}) - D_h(\tau^h_i+\Delta')\\
    && \;\;\;\;\;\;\;\;\;\;\;\; \mbox{for all $k > 0$ and $i < 0$},
\end{eqnarray*}
then
\begin{equation}
    (E_0)^c = \bigcup_{k >0, i < 0}  \hat{B}_{ik}.
\end{equation}

Let us fix a particular $n > 2\lambda_h\Delta' > 0$, and define:

\begin{eqnarray*}
    G_n &=& \mbox{event that} D_m(3n/\lambda_h + \zeta^h_m) = 0 \\
    && \;\;\;\;\;\;\;\;\;\;\;\; \mbox{for $m = -n, -n+1, \ldots, -1, 0, +1, \ldots, n-1,n$}
\end{eqnarray*}

Then 
\begin{eqnarray}
P(E_0 | U_0) & \ge & P(E_0|U_0,G_n)P(G_n|U_0) \nonumber\\
& = & \left ( 1 - P(\cup_{k>0,i<0} \hat{B}_{ik}|U_0,G_n) \right) P(G_n|U_0)\nonumber\\
& \ge & \left ( 1 - \sum_{k>0,i<0} P(\hat{B}_{ik}|U_0,G_n) \right) P(G_n|U_0) \nonumber\\
& \ge &  ( 1 - a_n - b_n) P(G_n|U_0) \label{eqn:up_bound}
\end{eqnarray}
where
\begin{eqnarray}
a_n & := & \sum_{(i,k): -n \le i < 0 < k \le n} P(\hat{B}_{ik}|U_0,G_n)\\
b_n & := & \sum_{(i,k):  i < -n \text{~or~} k > n }P(\hat{B}_{ik}|U_0,G_n).
\end{eqnarray}

\noindent Consider two cases:

\noindent
{\bf  Case 1:} $-n \le i <0 <  k \le n$:
\begin{eqnarray*}
P(\hat{B}_{ik}|U_0,G_n) &=& P(\hat{B}_{ik}|U_0, G_n,\sum_{m = i}^{k-1} R_m + \Delta' \leq 3n/\lambda_h) \\
&& \;+\;P(\sum_{m = i}^{k-1} R_m + \Delta' > 3n/\lambda_h |U_0,G_n) \\
& \leq & P(\sum_{m = i}^{k-1} R_m + \Delta' > 3n/\lambda_h |U_0,G_n) \\
& \leq & P(\sum_{m = i}^{k-1} R_m > 5n/(2\lambda_h) |U_0) \\
& \leq & P(\sum_{m = i}^{k-1} R_m  > 5n/(2\lambda_h))/P(U_0)  \\
& \leq & A_5 e^{-\gamma_1 n} 
\end{eqnarray*}
for some positive  constants $A_5, \gamma_1$ independent of $n,k,i$. The last inequality follows from the fact that $R_i$'s are iid 
exponential random variables of  mean $1/\lambda_h$. Summing these terms, we have:
\begin{eqnarray*}
a_n & = & \sum_{(i,k): -n \le i < 0 < k \le n} P(B_{ik}|U_0, G_n) \\
& \leq & \sum_{(i,k): -n \le i < 0 < k \le n} A_5 e^{-\alpha_1 n} : = \bar{a}_n,
\end{eqnarray*}
which is bounded and moreover $\bar{a}_n\rightarrow 0$ as $n \rightarrow \infty$.

\noindent {\bf  Case 2:} $k>n \text{~or~} i<-n$:

For $0<\varepsilon<1$, let us define event $W^{\varepsilon}_{ik}$ to be:
\begin{equation}
\label{eq:honest_growth}
    W^{\varepsilon}_{ik} = \mbox{event that $D_h(\zeta^h_{k-1}) - D_h(\zeta^h_i+\Delta') \geq (1-\varepsilon)\frac{r_h}{\lambda_h}(k-i-1)$}.
\end{equation}
Then we have
\begin{equation*}
    P(\hat{B}_{ik}|U_0,G_n) \leq P(\hat{B}_{ik}|U_0, G_n,W^{\varepsilon}_{ik}) + P({W^{\varepsilon}_{ik}}^c |U_0,G_n).
\end{equation*}

We first bound $P({W^{\varepsilon}_{ik}}^c |U_0,G_n)$:
\begin{eqnarray}
P({W^{\varepsilon}_{ik}}^c |U_0,G_n) &\leq& P({W^{\varepsilon}_{ik}}^c | \zeta^h_{k-1} - \zeta^h_i - \Delta'> \frac{k-i-1}{(1+\varepsilon)\lambda_h}) \nonumber \\
&& \;+\;P(\zeta^h_{k-1} - \zeta^h_i - \Delta' \leq \frac{k-i-1}{(1+\varepsilon)\lambda_h}) \nonumber \\
& \leq & P({W^{\varepsilon}_{ik}}^c | \zeta^h_{k-1} - \zeta^h_i - \Delta' > \frac{k-i-1}{(1+\varepsilon)\lambda_h}) +e^{-\Omega(\varepsilon^2 (k-i-1))} \nonumber \\
& \leq & e^{-\Omega(\varepsilon^4 (k-i-1))} + e^{-\Omega(\varepsilon^2 (k-i-1))}  \nonumber \\
& \leq & A_6 e^{-\gamma_2 (k-i-1)}
\label{eqn:prob_event}
\end{eqnarray}
for some positive  constants $A_6, \gamma_2$ independent of $n,k,i$, where the second inequality follows from the Erlang tail bound ( as $\zeta^h_{k-1} - \zeta^h_i$ is sum of IID exponentials due to time-warping) and the third inequality follows from Proposition \ref{prop:bound-1}.

Meanwhile, we have 
\begin{eqnarray*}
&~&P(\hat{B}_{ik}|U_0, G_n,W^{\varepsilon}_{ik})  \\
&\leq& P(D_i(\sum_{m = i}^{k-1}R_m + \Delta' + \zeta^h_i) \geq (1-\varepsilon)\frac{r_h}{\lambda_h}(k-i-1) | U_0,G_n,W^{\varepsilon}_{ik})  \\
& \leq & P(D_i(\sum_{m = i}^{k-1}R_m + \Delta' + \zeta^h_i) \geq (1-\varepsilon)\frac{r_h}{\lambda_h}(k-i-1) \\
&&\;\;\;\;\;\;| U_0,G_n,W^{\varepsilon}_{ik},\sum_{m = i}^{k-1}R_m + \Delta' \leq (k-i-1)\frac{r_h+\phi_c\lambda_a }{2\phi_c\lambda_a}\frac{1}{\lambda_h}) \\
&& \;+\;P(\sum_{m = i}^{k-1}R_m + \Delta'  > (k-i-1)\frac{r_h+\phi_c\lambda_a }{2\phi_c\lambda_a}\frac{1}{\lambda_h}|U_0,G_n,W^{\varepsilon}_{ik}) \\
&\stackrel{(a)}{\leq}& P(\sum_{m = i}^{k-1}R_m + \Delta' > (k-i-1)\frac{r_h+\phi_c\lambda_a }{2\phi_c\lambda_a}\frac{1}{\lambda_h}|U_0,G_n,W^{\varepsilon}_{ik})\\
&& \;+\; e^{-\theta_c^*(k-i-1)\frac{r_h+\phi_c\lambda_a }{2\phi_c\lambda_a}\frac{1}{\lambda_h} + \left(\frac{(1-\varepsilon)}{c}\frac{r_h}{\lambda_h}(k-i-1)-1\right)\Lambda_c(\theta_c^*)}g\left((k-i-1)\frac{r_h+\phi_c\lambda_a }{2\phi_c\lambda_a}\frac{1}{\lambda_h}\right)\\
&\stackrel{(b)}{=}& P(\sum_{m = i}^{k-1}R_m + \Delta' > (k-i-1)\frac{r_h+\phi_c\lambda_a }{2\phi_c\lambda_a}\frac{1}{\lambda_h}|U_0,G_n,W^{\varepsilon}_{ik})\\
&& \;+\; e^{-\theta_c^*\frac{k-i-1}{\lambda_h}\left[\frac{r_h+\phi_c\lambda_a }{2\phi_c\lambda_a}- (1-\varepsilon)\frac{r_h}{\phi_c \lambda_a}\right]}e^{-\Lambda_c(\theta_c^*)}g\left((k-i-1)\frac{r_h+\phi_c\lambda_a }{2\phi_c\lambda_a}\frac{1}{\lambda_h}\right)
\end{eqnarray*}
where $(a)$ follows from Lemma~\ref{lem:step4-lemma}, 
$(b)$ follows from $\frac{\Lambda_c(\theta^*_c)}{\theta^*_c} = \frac{1}{\lambda_a \eta_c} = \frac{c}{\phi_c\lambda_a}$. The first term can be bounded as:
\begin{eqnarray*}
&~&P(\sum_{m = i}^{k-1}R_m + \Delta' > (k-i-1)\frac{r_h+\phi_c\lambda_a}{2\phi_c\lambda_a }\frac{1}{\lambda_h}|U_0,G_n,W^{\varepsilon}_{ik})\\
&=& P(\sum_{m = i}^{k-1}R_m + \Delta' > (k-i-1)\frac{r_h+\phi_c\lambda_a}{2\phi_c\lambda_a}\frac{1}{\lambda_h}|U_0, W^{\varepsilon}_{ik}) \\
&\leq& P(\sum_{m = i}^{k-1}R_m + \Delta' > (k-i-1)\frac{r_h+\phi_c\lambda_a}{2\phi_c\lambda_a}\frac{1}{\lambda_h})/P(U_0, W^{\varepsilon}_{ik}) \\
&\leq& A_7 e^{-\gamma_3(k-i-1)}
\end{eqnarray*}
for some positive  constants $A_7, \gamma_3$ independent of $n,k,i$. The last inequality follows from the fact that $(r_h+\phi_c\lambda_a)/(2\phi_c\lambda_a) > 1$ and the $R_i$'s have mean $1/\lambda_h$, while $P(U_0, W^{\varepsilon}_{ik})$ is a event with high probability as we showed in (\ref{eqn:prob_event}).
Then we have 
\begin{eqnarray}
\label{eqn:B_hat_bound}
     &&P(\hat{B}_{ik}|U_0,G_n) \nonumber \\
     & \leq& A_6 e^{-\alpha_2 (k-i-1)} \nonumber \\
     &+& e^{-\theta^*_c(k-i-1)\frac{r_h(1-\varepsilon)}{\lambda_h\phi_c\lambda_a}\left[\frac{r_h + \phi_c \lambda_a}{2(1-\epsilon)r_h} - 1\right]}e^{-\Lambda_c(\theta_c^*)}g((k-i-1)\frac{r_h+\phi_c\lambda_a }{2\phi_c\lambda_a}\frac{1}{\lambda_h}) \nonumber \\
     &+& A_7 e^{-\gamma_3(k-i-1)}.
\end{eqnarray}
Summing these terms, we have:
\begin{eqnarray*}
b_n & = & \sum_{(i,k):  i<-n \text{~or~} k > n}P(\hat{B}_{ik}|U_0,G_n) \\
& \le  &  \sum_{ (i,k): i<-n \text{~or~} k > n}  [A_6 e^{-\alpha_2 (k-i-1)} \\
&& \;\;+\; e^{-\theta^*_c(k-i-1)\frac{r_h(1-\varepsilon)}{\lambda_h\phi_c\lambda_a}\left[\frac{r_h + \phi_c \lambda_a}{2(1-\epsilon)r_h} - 1\right]}e^{-\Lambda_c(\theta_c^*)}g\left((k-i-1)\frac{r_h+\phi_c\lambda_a }{2\phi_c\lambda_a}\frac{1}{\lambda_h}\right) \\
&&\;\;+\; A_7 e^{-\gamma_3(k-i-1)}] \\
&:=& \bar{b}_n
\end{eqnarray*}
Here, from (\ref{eq:bound-g-t}), $g(.) \rightarrow \frac{\lambda_a}{- \theta_c^{\star}}$ as $n \rightarrow  \infty$. Therefore, $\bar{b}_n$ is bounded and moreover $\bar{b}_n \rightarrow 0$ as $n \rightarrow \infty$ when we set $\varepsilon$ to be small enough such that $\frac{r_h+\phi_c\lambda_a }{2(1-\varepsilon)r_h}<1$.

Substituting these bounds in (\ref{eqn:up_bound}) we finally get:
\begin{equation}
    P(E_0|U_0) > [1- (\bar{a}_n+ \bar{b}_n)]P(G_n|U_0)
\end{equation}
By setting $n$ sufficiently large such that $\bar{a}_n$ and $\bar{b}_n$ are sufficiently small, we conclude that $P(\hat{E}_0)> 0$.

\subsection{Proof of Lemma~\ref{lem:step5-lemma-2}}
\label{sec:B_s_s_plus_t}

We divide the proof in to two steps. In the first step, we prove for $\varepsilon = 1/2$.
Recall that we have defined event $\hat{B}_{ik}$ as:
\begin{equation*}
    \hat{B}_{ik} = \mbox{event that $D_i(\sum_{m = i}^{k-1} R_m + \Delta' + \zeta^h_i)  \ge  D_h(\zeta^h_{k-1}) - D_h(\zeta^h_i+\Delta')$}.
\end{equation*}


And by Lemma~\ref{lem:event-reduction-1},~\ref{lem:event-reduction-2},~\ref{lem:F_j},  we have
\begin{equation}
    \hat{F}_j^c = F_j^c \cup V_j^c = \left(\bigcup_{(i,k): i < j <k} \hat{B}_{ik}\right) \cup V_j^c.
\end{equation}

For $t' > \max\left\{\left(\frac{2\lambda_h}{1-\eta}\right)^2\left(\frac{c-1}{\phi_c-1}\right)^2, \left[(c-1)\left(\Delta' + \frac{1}{\lambda_{\min}}\right)\right]^2\right\}$, we have $\frac{\sqrt{t'}}{2\lambda_h} > \frac{\lambda_h}{\lambda_a}\left(\frac{c-1}{\phi_c-1}\right)$ and $\sqrt{t'} > (c-1)\left(\Delta' + \frac{1}{\lambda_{\min}}\right)$.

Divide $[s',s'+t']$ into $\sqrt{t'}$ sub-intervals of length $\sqrt{t'}$, so that the $r$ th sub-interval is:
$$\J_r : = [s'+  (r-1) \sqrt{t'}, s'+ r\sqrt{t'}].$$
Now look at the first, fourth, seventh, etc sub-intervals, i.e. all the $r = 1 \mod 3$ sub-intervals. Introduce the event that in the $\ell$-th $1 \mod 3$th sub-interval,
an adversary tree that is rooted at a honest block arriving in that sub-interval or in the previous ($0 \mod 3$) sub-interval catches up with a honest block in that sub-interval or in the next ($2 \mod 3$) sub-interval. 
Formally,
$$C_{\ell}=\bigcap_{j: \zeta^h_j \in \J_{3\ell+1}}
U_j^c \cup \left(\bigcup_{(i,k): \zeta^h_j - \sqrt{t'} < \zeta^h_i < \zeta^h_j, \zeta^h_j < \zeta^h_k +\Delta' < \zeta^h_j +\sqrt{t'} } \hat{B}_{ik} \right).$$
We have
\begin{equation}
    \label{eqn:qq2}
    P(C_{\ell})\leq P(\mbox{no arrival in $\J_{3\ell+1}$}) + 1-p < 1
\end{equation}
for large enough $t'$, where $p$ is a uniform lower bound such that $P(\hat{F}_j) \ge p$ for all $j$. 
Also, we define the following event:
$$
\hat{C}_\ell = \text{ event that the honest fictitious tree grows by $c-1$ levels in sub-interval } \J_{3\ell+2}
$$
Observe that because of $\rs$ being updated at each epoch beginning, for distinct $\ell$, the events $C_\ell \bigcap \hat{C}_\ell$  are independent. Using Chernoff bounds, for $\sqrt{t'} > (c-1)\left(\Delta' + \frac{1}{\lambda_{\min}}\right)$, we have $P(\hat{C}_\ell) \geq 1 - e^{-c_2\sqrt{t'}}$.

Introduce the atypical events:
\begin{eqnarray}
    B &=& \bigcup_{(i,k): \zeta^h_i \in [s',s'+t'] \mbox{~or~} \zeta^h_k + \Delta' \in [s',s'+t'], i < k, \zeta^h_k + \Delta' - \zeta^h_i >  \sqrt{t'}} \hat{B}_{ik}, \nonumber \\ \text{ and } \nonumber\\
    \tilde{B} &=& \bigcup_{(i,k):\zeta^h_i<s',s'+t'<\zeta^h_k+\Delta'} \hat{B}_{ik}. \nonumber
\end{eqnarray}
The events $B$ and $\tilde{B}$ are superset of the events  that an adversary tree catches up with an honest block far ahead. Then we have
\begin{align}
\label{eqn:qqq2}  & P(B^{static}_{s',s'+t'}) \leq P(\bigcap_{j: \zeta^h_j \in [s',s'+t']} U_j^c) + P(B)+P(\tilde B)+P(\bigcap_{\ell=0}^{\sqrt{t'}/3} C_{\ell})\nonumber \\
&\quad\quad \leq  P(\bigcap_{j: \zeta^h_j \in [s',s'+t']} U_j^c) + P(B)+P(\tilde B)+ P(\bigcup_{\ell=0}^{\sqrt{t'}/3} \hat{C}^c_\ell) + P(\bigcap_{\ell=0}^{\sqrt{t'}/3} C_{\ell}\cap \hat{C}_\ell)\nonumber \\
&\quad\quad \leq P(\bigcap_{j: \zeta^h_j \in [s',s'+t']} U_j^c) + P(B)+ P(\tilde B)  +  \sum_{\ell=0}^{\sqrt{t'}/3}P(\hat{C}^c_\ell) +(P(C_{\ell}\cap \hat{C}_\ell))^{\sqrt{t'}/3}\nonumber \\
&\quad\quad \leq  e^{-c_1 t'}+ P(B)+ P(\tilde B) + e^{-c_2 \sqrt{t'}} + (P(C_\ell))^{\frac{\sqrt{t'}}{3}}
\end{align}
for some positive constants $c_1, c_2$ when $t'$ is large. Next we will bound the atypical events $B$ and $\tilde{B}$. Consider the following events
\begin{eqnarray*}
D_1&=&\{\#\{i: \zeta^h_i\in (s'-\sqrt{t'}-\Delta',s'+t'+\sqrt{t'}+\Delta)\} >2\lambda_h t'\} \label{eqn:D1}\\
D_2&=&\{ \exists i,k: \zeta^h_i  \in (s',s'+t'), (k-i)<\frac{\sqrt{t'}}{2\lambda_h}, \zeta^h_k-\zeta^h_i+\Delta'>\sqrt{t'}\} \label{eqn:D2}\\
D_3&=&\{ \exists i,k: \zeta^h_k+\Delta \in (s',s'+t'), (k-i)<\frac{\sqrt{t'}}{2\lambda_h}, \zeta^h_k-\zeta^h_i+\Delta'>\sqrt{t}\} \label{eqn:D3}
\end{eqnarray*}
In words, $D_1$ is the event of atypically many honest arrivals in $(s'-\sqrt{t'}-\Delta',s'+t'+\sqrt{t'}+\Delta')$ while $D_2$ and $D_3$ are the events that there exists an interval of length $\sqrt{t'}$ with at least one endpoint inside $(s',s'+t')$ with atypically small number of arrivals. Since, by time-warping, the number of honest arrivals in $(s',s'+t')$ (in the local clock of the static system) is Poisson with parameter $\lambda_h t'$, we have from the memoryless property of the Poisson process that  $P(D_1)\leq e^{-c_0t'}$ for some constant $c_0=c_0(\lambda_a,\lambda_h)>0$ when $t'$ is large.  
On the other hand, using the memoryless property and a union bound, and decreasing $c_0$ if needed, we have that $P(D_2)\leq e^{-c_0 \sqrt{t'}}$. Similarly, using time reversal, $P(D_3)\leq e^{-c_0\sqrt{t'}}$.
Therefore, again using the memoryless property of the Poisson process,
\begin{eqnarray}
P(B)&\leq & P(D_1\cup D_2\cup D_3)+ P(B\cap D_1^c\cap D_2^c\cap D_3^c)\nonumber\\
&\leq & e^{-c_0 t'} + 2e^{-c_0\sqrt{t'}}+\sum_{i=1}^{2\lambda_h t'} \sum_{k: k-i>\sqrt{t'}/2\lambda_h} P(\hat{B}_{ik}) \\
&\leq & e^{-c_3\sqrt{t'}},
\label{eqn:PB}
\end{eqnarray}
for large $t'$, where $c_3>0$ are constants that may depend on $\lambda_a,\lambda_h$ and the last inequality is due to (\ref{eqn:PBik}). 
We next claim that there exists a constant $\alpha>0$ such that, for all $t'$ large,
\begin{equation}
    \label{eqn:PtB}
    P(\tilde B)\leq e^{- c_6 t'}.
\end{equation}
Consider the following event
$$
D_4 = \{ \exists i,k: \zeta^h_i<s',s'+t'<\zeta^h_k+\Delta', (k-i)<\frac{t'}{2\lambda_h}, \zeta^h_k-\zeta^h_i+\Delta'>t'\}.
$$
Using Poisson tail bounds, we can show that $P(D_4) \leq e^{-c_4t'}$. Now, we have 
\begin{eqnarray}
&&P(\tilde B) \leq P(D_4) + P(\tilde B \cap D^c_4)\nonumber\\
&\leq& e^{-c_4 t'} + \sum_{i,k:k-i >t'/2\lambda_h} \int_0^{s'} P(\zeta^h_i\in d\theta) P(\hat{B}_{ik}, \zeta^h_k-\zeta^h_i+\Delta'>s' + t'-\theta)\nonumber \\
&\leq & e^{-c_4 t'} + \sum_i \int_0^{s'} P(\zeta^h_i\in d\theta) \sum_{k:k-i>t'/2\lambda_h} P(\hat{B}_{ik})^{1/2} P(\zeta^h_k-\zeta^h_i+\Delta'>s' + t'-\theta)^{1/2}.\nonumber\\
\label{eqn:PtB1}
\end{eqnarray}
The tails of the Poisson distribution yield the existence of constants $c',c''>0$ so that
\begin{eqnarray}
    \label{eqn:Poisson_tail}
    &&P(\zeta^h_k-\zeta^h_i + \Delta'>s'+t'-\theta)\\
    &\leq& \left\{
    \begin{array}{ll}
    1,& (k-i)>c'(s'+t'-\theta-\Delta')\\
    e^{-c''(s'+t'-\theta-\Delta')},& (k-i)\leq c'(s'+t'-\theta-\Delta').
    \end{array}\right.
\end{eqnarray} 
(\ref{eqn:PBik}) and \eqref{eqn:Poisson_tail} yield that, for large enough $t'$, there exists a constant $c_5>0$ so that
\begin{equation}
    \label{eqn:PtB2}
    \sum_{k: k-i > t'/2\lambda_h} P(\hat{B}_{i,k})^{1/2}P(\zeta^h_k-\zeta^h_i>s'+t'-\theta-\Delta')^{1/2} \leq e^{-2c_5(s'+t'-\theta-\Delta')}.
\end{equation}
Substituting this bound in \eqref{eqn:PtB1} and using that $\sum_i P(\zeta^h_i\in d\theta)=d\theta$ gives
\begin{eqnarray}
\label{eqn:PtB3}
P(\tilde B)&\leq & e^{-c_4 t'} + 
\sum_{i} \int_0^{s'} 
P(\zeta^h_i\in d\theta) e^{-2c_5 (s'+t'-\theta-\Delta')}\nonumber\\
&\leq& e^{-c_4 t'} +  \int_0^{s'} e^{-2c_5(s'+t'-\theta-\Delta')} d\theta
\leq e^{-c_4 t'} + \frac{1}{2c_5} e^{-2c_5(t'-\Delta')} \nonumber \\
&\leq& e^{-c_6 t'},
\end{eqnarray}
for $t'$ large and $c_6 = \min(c_4,c_5)$, proving \eqref{eqn:PtB}.

Combining (\ref{eqn:PB}), (\ref{eqn:PtB3}) and (\ref{eqn:qqq2}) concludes the proof of step 1.

In step two, we prove for any $\varepsilon > 0$ by recursively applying the bootstrapping procedure in step 1.
Assume the following statement is true: for any $\theta \geq m$ there exist constants $\bar b_\theta,\bar A_\theta$ so that for all $s',t'\geq 0$,
\begin{equation}
\label{eqn:qst_basic}
\tilde{q}[s',s'+t'] \leq \bar A_\theta \exp(-\bar b_\theta t'^{1/\theta}).
\end{equation}
By step 1, it holds for $m = 2$. Also, for specific values of $m$ that 
we will consider, we will have $t'^{\frac{m}{2m-1}} > \sqrt{t'}$.

Divide $[s',s'+t']$ into $t'^{\frac{m-1}{2m-1}}$ sub-intervals of length $t'^{\frac{m}{2m-1}}$, so that the $r$ th sub-interval is:
$$\J_r : = [s'+  (r-1) t'^{\frac{m}{2m-1}}, s'+ rt'^{\frac{m}{2m-1}}].$$

Now look at the first, fourth, seventh, etc sub-intervals, i.e. all the $r = 1 \mod 3$ sub-intervals. Introduce the event that in the $\ell$-th $1 \mod 3$th sub-interval,
an adversary tree that is rooted at a honest block arriving in that sub-interval or in the previous ($0 \mod 3$) sub-interval catches up with a honest block in that sub-interval or in the next ($2 \mod 3$) sub-interval. 
Formally,
$$C_{\ell}=\bigcap_{j: \zeta^h_j \in \J_{3\ell+1}}
U_j^c \cup \left(\bigcup_{(i,k): \zeta^h_j - t'^{\frac{m}{2m-1}} < \zeta^h_i < \zeta^h_j, \zeta^h_j < \zeta^h_k +\Delta' < \zeta^h_j +t'^{\frac{m}{2m-1}} } \hat{B}_{ik} \right).$$
By (\ref{eqn:qst_basic}), we have
\begin{equation}
    \label{eqn:pcl}
    P(C_{\ell})\leq  A_m \exp(-\bar a_m t'^{\frac{1}{2m-1}}).
\end{equation}

Also, we define the following event:
$$
\hat{C}_\ell = \text{ event that the honest fictitious tree grows by $c-1$ levels in sub-interval } \J_{3\ell+2}
$$
Note that for distinct $\ell$, the events $C_\ell \bigcap \hat{C}_\ell$  are independent. Also, from Lemma~\ref{lem:step1-lemma-1}, assuming $t'^{\frac{m}{2m-1}} > \sqrt{t'} > (c-1)\left(\Delta' + \frac{1}{\lambda_{\min}}\right)$, we have $P(\hat{C}_\ell) \geq 1 - e^{-c_2t'^{\frac{m}{2m-1}}}$ for some positive constant $c_2$.

Introduce the atypical events:
\begin{eqnarray}
    B &=& \bigcup_{(i,k): \zeta^h_i \in [s',s'+t'] \mbox{~or~} \zeta^h_k + \Delta' \in [s',s'+t'], i < k, \zeta^h_k + \Delta' - \zeta^h_i >  t'^{\frac{m}{2m-1}}} \hat{B}_{ik}, \nonumber \\ \text{ and } \nonumber\\
    \tilde{B} &=& \bigcup_{(i,k):\zeta^h_i<s',s'+t'<\zeta^h_k+\Delta'} \hat{B}_{ik}.  \nonumber
\end{eqnarray}
The events $B$ and $\tilde{B}$ are the events  that an adversary tree catches up with an honest block far ahead. Following the calculations in step 1, we have
\begin{eqnarray}
    P(B) &\leq& e^{-c_3 t'^{\frac{m}{2m-1}}}\\
    P(\tilde{B}) &\leq& e^{- c_6 t'},
\end{eqnarray}
for large $t'$, where $c_1$ and $c_5$ are some positive constant.

Then we have
\begin{align*}
\label{eqn:qst_induc}
& \tilde{q}[s',s'+t'] \leq P(\bigcap_{j: \zeta^h_j \in [s',s'+t']} U_j^c) + P(B)+P(\tilde B)+P(\bigcap_{\ell=0}^{t'^{\frac{m-1}{2m-1}}/3} C_{\ell})\nonumber \\
&\quad \leq P(\bigcap_{j: \zeta^h_j \in [s',s'+t']} U_j^c) + P(B)+P(\tilde B)+ P(\bigcup_{\ell=0}^{t'^{\frac{m-1}{2m-1}}/3} \hat{C}^c_\ell) + P(\bigcap_{\ell=0}^{t'^{\frac{m-1}{2m-1}}/3} C_{\ell} \cap \hat{C}_\ell)\nonumber \\ 
&\quad \leq P(\bigcap_{j: \zeta^h_j \in [s,s+t]} U_j^c) + P(B)+ P(\tilde B) + \sum_{\ell=0}^{t'^{\frac{m-1}{2m-1}}/3} P(\hat{C}^c_\ell) + (P(C_{\ell} \cap \hat{C}_\ell))^{t'^{\frac{m-1}{2m-1}}/3}\nonumber \\
&\quad \leq  e^{-c_1 t'}+ e^{-c_3 t'^{\frac{m}{2m-1}}}+ e^{- c_6 t'} +  e^{-c_2t'^{\frac{m}{2m-1}}}  + ( A_m \exp(-\bar a_m t'^{1/(2m-1)}))^{t'^{\frac{m-1}{2m-1}}/3} \nonumber \\
&\quad \leq \bar A'_m \exp(-\bar b'_m t'^{\frac{m}{2m-1}})
\end{align*}
for large $t'$, where $A'_m$ and $b'_m$ are some positive constant.

So we know the statement in (\ref{eqn:qst_basic}) holds for all $\theta \geq \frac{2m-1}{m}$. Start with $m_1=2$, we have a recursion equation $m_k = \frac{2m_{k-1}-1}{m_{k-1}}$ and we know (\ref{eqn:qst_basic}) holds for all $\theta \geq m_k$. It is not hard to see that $m_k = \frac{k+1}{k}$ and thus $\lim_{k\rightarrow\infty} m_k = 1$.  Now observe that for $m_k = \frac{k+1}{k}$, we have $t'^{\frac{m_k}{2m_k-1}} > \sqrt{t'}$ for $k > 1$. 

So, for some constant $\bar a_\theta$ which is a function of $\Delta'$, we can rewrite (\ref{eqn:qst_basic}) as
\begin{align*}
    \tilde{q}[\alpha(s),\alpha(s+t)] \leq \bar A'_m \exp(-\bar a_\theta t^{1/\theta})
\end{align*}
which concludes the lemma.

\section{Proof of Lemma~\ref{lem:step6-lemma}}
\label{sec:step6-lemma-proof}
Let $U_j$ be the event in $ss1$ that the $j-$th honest block $b_j$
is a loner, i.e.,
\begin{align*}
    U_j = \{\tau^h_{j-1} < \tau^h_{j} - \Delta\} \bigcap \{\tau^h_{j+1} > \tau^h_{j} + \Delta\}
\end{align*}
\noindent Let $\hat{F}_j = U_j \bigcap F_j$ be the event that $b_j$ is a Nakamoto block. 
We define the following ``potential" catch up event in $ss1$: 
\begin{equation}
    \hat{A}_{ik} = \{D_i(\alpha(\tau^h_{k} + \Delta))  \ge  D_h(\alpha(\tau^h_{k-1})) - D_h(\alpha(\tau^h_i+\Delta))\},
\end{equation}
which is the event that the adversary launches a private attack starting from honest block $b_i$ and catches up the fictitious honest chain right before honest block $b_k$ is proposed. 

Next, define the following events
\begin{align}
    V^{ss1}_j &= \{\alpha(\tau^h_{j-1}) < \alpha(\tau^h_{j}) - \frac{\lambda_{\max}}{\lambda_h}\Delta\} \bigcap \{\alpha(\tau^h_{j+1}) > \alpha(\tau^h_{j}) + \frac{\lambda_{\max}}{\lambda_h}\Delta\}\\
    \hat{B}^{ss1}_{ik} &=  \{D_i(\alpha(\tau^h_k) + \frac{\lambda_{\max}}{\lambda_h}\Delta)  \ge  D_h(\alpha(\tau^h_{k-1})) - D_h(\alpha(\tau^h_i) + \frac{\lambda_{\max}}{\lambda_h}\Delta)\}
\end{align}

\begin{lemma}
\label{lem:event-reduction-1}
For any pair of $i,k$,
$$
\hat{A}_{ik} \subseteq \hat{B}^{ss1}_{ik}. 
$$
\end{lemma}
\begin{proof}
Using equation~\ref{eq:time-relations}, we have
\begin{align*}
    \alpha(\tau^h_k + \Delta) &= \int_{0}^{\tau^h_k + \Delta}\frac{\lambda^c_h(u)}{\lambda_h}du = \int_{0}^{\tau^h_k}\frac{\lambda^c_h(u)}{\lambda_h}du + \int_{\tau^h_k}^{\tau^h_k + \Delta}\frac{\lambda^c_h(u)}{\lambda_h}du\\
    &\leq \alpha(\tau^h_k) + \frac{\lambda_{\max}}{\lambda_h}\Delta
\end{align*}
Similarly, $\alpha(\tau^h_i + \Delta) \leq \alpha(\tau^h_i) + \frac{\lambda_{\max}}{\lambda_h}\Delta$. Because $D_h(.)$ and $D_i(.)$ are increasing functions over their domain, we have 
\begin{align*}
    D_i(\alpha(\tau^h_k + \Delta)) &\leq D_i(\alpha(\tau^h_k) + \frac{\lambda_{\max}}{\lambda_h}\Delta) \text{ and }\\
    D_h(\alpha(\tau^h_i + \Delta)) &\leq D_h(\alpha(\tau^h_i) + \frac{\lambda_{\max}}{\lambda_h}\Delta)
\end{align*}
\end{proof}

\begin{lemma}
\label{lem:event-reduction-2}
For all $j$,
$$V^{ss1}_j \subseteq U_j.$$ 
\end{lemma}
\begin{proof}
    This can be proved using the fact that $\int_{\tau^h_{j-1}}^{\tau^h_{j-1} + \Delta} \frac{\lambda^c_h(u)}{\lambda_h}du \leq \frac{\lambda_{\max}}{\lambda_h}\Delta$ and $\int_{\tau^h_{j}}^{\tau^h_{j} + \Delta} \frac{\lambda^c_h(u)}{\lambda_h}du \leq \frac{\lambda_{\max}}{\lambda_h}\Delta$.
\end{proof}

By time-warping, $R_m$ is an IID exponential random variable with rate $\lambda_h$. Let $\zeta^h_j = \alpha(\tau^h_j)$, that is, $\zeta^h_j$ is the time of mining of $j-$th honest block in the local clock of static system $ss1$. Similarly, we define $\zeta^a_j = \alpha(\tau^a_j)$ for the $j-$th adversarial block. Then, we can rewrite the event $\hat{B}_{ik}$ as:
\begin{align}
    \hat{B}^{ss1}_{ik} &=  \Bigg\{D_i(\zeta^h_k + \frac{\lambda_{\max}}{\lambda_h}\Delta)  \ge  D_h(\zeta^h_{k-1}) - D_h(\zeta^h_i +\frac{\lambda_{\max}}{\lambda_h}\Delta)\Bigg\}. \nonumber
\end{align}

\begin{lemma}  
\label{lem:ss0_F_j}
In the static system $ss1$, for each $j$
\begin{equation}
\label{eq:ss0_F_j}
    P(\hat{F}_j^c) = P(F_j^c \cup U_j^c) \leq P\left(\left(\bigcup_{(i,k): 0 \leq i < j <k} \hat{B}^{ss1}_{ik}\right) \cup (V^{ss1}_j)^c\right).
\end{equation}
\end{lemma}
This can be proved in a similar way as Lemma~\ref{lem:F_j} and using Lemma~\ref{lem:event-reduction-1},~ \ref{lem:event-reduction-2}. Furthermore, defining $X_d$, $d > 0$, as the time it takes in the local clock of static system $ss1$ for $D_h$ to reach depth $d$ after reaching depth $d-1$, we have
\begin{proposition}
\label{prop:ss1-bound-1}
Let $Y_d$, $d \geq 1$, be i.i.d random variables, exponentially distributed with rate $\lambda_h$. Then, each random variable $X_d$ is less than $\Delta'+Y_d$, where $\Delta' = \frac{\lambda_{\max}}{\lambda_h}\Delta$.
\end{proposition}
\begin{proof}
Let $h_i$ be the first block that comes at some depth $d-1$ within $\T_h$. Then, in the local clock of static system, every honest block that arrives within interval $[\alpha(\tau^h_i),\alpha(\tau^h_i +\Delta)]$ will be mapped to the same depth as $h_i$, i.e., $d-1$. Hence, $\T_h$ will reach depth $d$ only when an honest block arrives after time $\alpha(\tau^h_i +\Delta)$. Now, due to time warping, in the local clock of static system $ss1$, we know that the difference between $\alpha(\tau^h_i +\Delta)$ and the arrival time of the first honest block after $\alpha(\tau^h_i +\Delta)$ is exponentially distributed with rate $\lambda_h$ due to the memoryless property of the
exponential distribution. This implies that for each depth d, $X_d = \alpha(\tau^h_i +\Delta) - \alpha(\tau^h_i) + Y_d = \int_{\tau^h_i}^{\tau^h_i + \Delta}\frac{\lambda^c_h(u)}{\lambda_h}du + Y_d \leq \Delta' + Y_d$ for some random variable $Y_d$ such that $Y_d, d\geq 1,$ are IID and exponentially distributed with rate $\lambda_h$.
\end{proof}
Thus, for $\Delta' = \frac{\lambda_{\max}}{\lambda_h}\Delta$, Proposition~\ref{prop:ss1-bound-1} implies that both Proposition~\ref{prop:bound-1} and Proposition~\ref{prop:bound-3} are satisfied for the static system $ss1$. Therefore,  for $\Delta' = \frac{\lambda_{\max}}{\lambda_h}\Delta$, a similar result holds for the event $\hat{B}^{ss1}_{ik}$ as in Lemma~\ref{lem:PBik_pow}. Additionally, Lemma~\ref{lem:step4-lemma} holds for $ss1$. Then, substituting $\Delta' = \frac{\lambda_{\max}}{\lambda_h}\Delta$ and using Lemma~\ref{lem:ss0_F_j}, we have both Lemma~\ref{lem:step5-lemma-1} and Lemma~\ref{lem:step5-lemma-2} satisfy for the static system $ss1$. For a time $t > 0$ in the local clock of the dynamic available system $dyn2$, we have $\alpha(t) \geq \frac{\lambda_{\min}}{\lambda_h}t$. Then, using Lemma~\ref{lem:step1-lemma}, Lemma~\ref{lem:step2-lemma}, Lemma~\ref{lem:step3-lemma}, we conclude the proof.











\end{document}